%% file: main.tex
\documentclass{article}

\usepackage{arxiv}

\usepackage{graphicx}
\usepackage{hyperref}
\usepackage{amsmath, amssymb}
\usepackage{amsfonts}
\usepackage{amsthm}
\usepackage{dsfont}
\usepackage{mathrsfs}
\usepackage{algpseudocode}
\usepackage{subcaption}
\usepackage{multirow}
\usepackage{soul}
\usepackage{float}
\usepackage{listings}
\DeclareCaptionStyle{ruled}{labelfont=normalfont,labelsep=colon,strut=off} 
\floatstyle{ruled}
\newfloat{listing}{tb}{lst}{}
\floatname{listing}{Listing}

\usepackage[cal=rsfs]{mathalfa}
\usepackage[table]{xcolor}
\definecolor{lightgray}{rgb}{0.83, 0.83, 0.83}

\newtheorem{definition}{Definition}
\newtheorem{theorem}{Theorem}
\newtheorem{lemma}{Lemma}
\newtheorem{proposition}{Proposition}
\newtheorem{remark}{Remark}

\newcommand{\ourmech}{\texttt{2-Prop}}
\newcommand{\defi}{\texttt{DeFi}}

\newcommand{\mev}{\texttt{MEV}}
\newcommand{\curmech}{$\Xi$}
\newcommand{\thd}{\textsf{K}}

\title{Let Leaders Play Games: Improving Timing in Leader-based Consensus}

\author{
 Rasheed M\\
  Machine Learning Lab\\
  IIIT Hyderabad\\
  Hyderabad, India\\
  \texttt{mohammad.ahmed@research.iiit.ac.in}\\
   \And
 Parth Desai \\
  Machine Learning Lab\\
  IIIT Hyderabad\\
  Hyderabad, India\\
  \texttt{parth.desai@research.iiit.ac.in}\\
  \And
Sujit Gujar \\
  Machine Learning Lab\\
  IIIT Hyderabad\\
  Hyderabad, India\\
  \texttt{sujit.gujar@iiit.ac.in}\\
}

\begin{document}
\maketitle
\begin{abstract}
Propagation latency is inherent to any distributed network, including blockchains. Typically, blockchain protocols provide a timing buffer for block propagation across the network. In leader-based blockchains, the leader -- block \emph{proposer} --  is known in advance for each slot. A fast (or low-latency) proposer may delay the block proposal in anticipation of more rewards from the transactions that would otherwise be included in the subsequent block. Deploying such a strategy by manipulating the timing is known as \emph{timing games}. It increases the risk of missed blocks due to reduced time for other nodes to vote on the block, affecting the overall efficiency of the blockchain.
 Moreover, proposers who play timing games essentially appropriate \mev\ (additional rewards over transaction fees and the block reward) that would otherwise accrue to the next block, making it unfair to subsequent block proposers. We propose a double-block proposal mechanism, \ourmech\, to curtail timing games. \ourmech\ selects two proposers per slot to propose blocks and confirms one of them. We design a reward-sharing policy for proposers based on how quickly their blocks propagate to avoid strategic deviations. In the induced game, which we call the \emph{Latency Game}, we show that it is a Nash Equilibrium for the proposers to propose the block without delay under homogeneous network settings. Under heterogeneous network settings, we study many configurations, and our analysis shows that a faster proposer would prefer not to delay unless the other proposer is extremely slow. Thus, we show the efficacy of \ourmech\ in mitigating the effect of timing games. 
\end{abstract}

\section{Introduction}
\label{sec:intro}

\textbf{Consensus Protocols and Blockchains.}
\label{ssec:con_proto}
\emph{Blockchain} is a distributed database that stores data in a sequence of blocks, primarily consisting of transactions. A distributed collection of nodes, known as miners or validators, maintains the blockchain database. \emph{Consensus protocols} are critical in distributed systems, as they ensure all participating nodes agree on a single, consistent state despite failures or malicious behavior. This agreement is crucial in blockchains, where a distributed network must validate and append new data without relying on a central authority/administrator. Nakamoto proposed Bitcoin~\cite{nakamoto2008bitcoin} based on the \emph{Proof-of-Work} (PoW) mechanism, in which nodes must solve a cryptographic puzzle to propose a new block; the process of solving the puzzle is called \emph{mining}. The probability of mining the next block is proportional to the computational power of the \emph{miner} (node). Following Bitcoin's success, more energy-efficient \emph{leader-based} blockchain protocols have been proposed.

Leader-based consensus protocols progress in rounds, also known as \emph{slots}. In each slot, the protocol selects a committee of nodes, referred to as \emph{validators}, to add a new block to the ledger. The node that proposes the next block is referred to as \emph{proposer}, and the other nodes that validate and attest to the proposed block are referred to as \emph{attestors}. Popular consensus protocols include  \emph{Proof-of-Stake} (PoS) \cite{buterin2013ethereum, buchman2016tendermint, ourborous, ouroborosGenesis, ouroborosPraos, yin2019hotstuff}, Proof-of-Space~\cite{ProofofSpace}, Proof-of-Reputation~\cite{repuCoin}. In these protocols, all the validators, particularly the proposer for the next slot, are known in advance. These protocols specify the expected progression within each slot, along with rewards for both proposers and attestors. Typically, the proposer gathers transactions after the previous block is proposed and publishes its block at the beginning of the next slot, $t=0$, so that it reaches all the attestors by the attestation deadline, $t=\tau_1$. Attestors then submit their attestations, which are aggregated by the aggregation deadline ($t=\tau_2$), and the block is confirmed/rejected at the end of the slot, \(t=\tau\) (see Figure~\ref{fig:slotProgEth}). Upon confirmation, the proposer receives transaction fees and block rewards, while attestors receive rewards for correct attestations, i.e., those that align with the majority.

Within the blockchain ecosystem, \emph{Decentralized Finance} (\defi) has reshaped traditional financial structures. A key phenomenon in this transformation was \emph{Maximal Extractable Value} (\mev), which refers to the maximum value that can be extracted from block production over the standard block reward and transaction fees by including, excluding, and reordering of transactions in a block \cite{ethereum2024b}. 
\mev\ thrives in environments with any DeFi activity and has profound implications on decentralization \cite{sokMEV, RegMEV}. 

A strategic proposer with access to a faster network can delay the block proposal until $0<\delta<\tau_1$, in anticipation of capturing more transaction fees and extracting more \mev\ from the additional transactions from $t=0$ to $t=\delta$. Deploying such a strategy by manipulating the timing of block proposal is known as \emph{timing games}~\cite{timingTimingpics}. The concept of a timing game is prevalent in the Ethereum blockchain~\cite{BurakTiming2023}, and therefore, we outline the main ideas of this paper in this context. However, our results hold for any leader-based blockchain protocol where the leader for each slot is known in advance. 

\textbf{Criticality of Resolving Timing Games.}\label{ssec:timing_game_imp}
The emergence of \mev\ highlighted how timing games can reduce the time for attestation and result in an increased risk of missed slots, affecting the overall health and efficiency of the network 
\cite{wahrstatter2023timebribemeasuringblock}. In addition, honest proposers can find the optimal block (or \mev\ bundles) from high-valued transactions available only after $\tau-\delta$ time (as the previous proposer would have captured such transactions). Thus, it may contain fewer or less profitable \mev\ opportunities. Essentially, the proposers playing timing games steal the transaction fees/\mev\ that the next proposer could have earned. This incentive misalignment creates a centralizing pressure, as only sophisticated validators, such as those connected to relay services (more on this in Section \ref{sec:rel_work}) with fast networks, can avoid losses. Due to this unfair environment, honest proposers may drop out of the system. Thus, the timing game is a serious concern for fairness and centralization. This paper focuses on resolving timing games in leader-based consensus protocols for blockchains. 

\begin{figure}[t]
    \centering
    \includegraphics[scale=1.8]{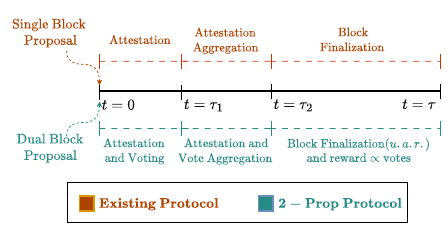}
    \caption{Progression within Slot}
    \label{fig:slotProgEth}
\end{figure}

\textbf{Challenges.}\label{ssec:timing_game_chall}
If all proposers shift to a fast network and delay the block proposal by some $\delta_{new}$, the result is a new system where each event is right-shifted by $\delta_{new}$. Proposers in the new system can further delay close to the new attestation deadline, resulting in the same problem. 

One may reasonably expect attestors to stop timing games. However, there are two problems with it: (i) An attestor cannot distinguish intentional proposer delays from natural network delays; (ii) Even if attestors are highly confident in the timing of the block proposal, it is a Nash equilibrium (NE) for them to attest for a block received within $\tau_1$. Attestors receive rewards only when their attestation aligns with the majority threshold attestations and gets confirmed on the blockchain \cite{eth2bookUpgradingEthereum}. Alternatively, the fees from one proposer can be given to the proposer of the next block. This approach has adverse implications for the consensus. The proposer can always produce empty blocks irrespective of the rewards earned from the previous block, affecting throughput. Thus, there is a need for a better solution to curtail timing games. 

\textbf{Our Approach.}
\label{ssec:our_approach}
The challenge in addressing timing games lies in distinguishing intentional proposer delays from natural network delays. To enforce early block publication, we introduce a double block-proposal mechanism, which we refer to as \ourmech. The protocol selects two proposers per slot and allows each to propose a block. Similar to \curmech, the attestors sign a valid block within $\tau_1$, and a block is confirmed only when it achieves at least a threshold number of votes ($\thd$). The issues to consider are: (i) Which block should the attestors attest to? (ii) Which blocks should be confirmed? and (iii) Who should get the rewards? In other words, how should incentives be designed?

\ourmech\ uses randomization when both blocks obtain \thd\ attestations and uses a reward-sharing policy that incentivizes both fast and slow proposers to mitigate the effect of timing games. We analyze \ourmech\ in a game-theoretic framework. Specifically, we model it as a two-player game, \emph{Latency Game}, where the strategies are the amounts of delay in the block proposal. We derive a closed-form solution for their expected utilities as functions of block delays and network parameters. We show that in the \emph{homogeneous proposer} setting (both the proposers have the same network parameters), it is a \emph{Nash Equilibrium} (NE) not to delay the proposal. In the \emph{heterogeneous proposers} setting, analytically computing NE for general network parameters is challenging. Hence, we model it as a bi-matrix game with a discretized strategy space ($\delta$s here). We analyze such a bi-matrix game for many network configurations. Our analysis shows that a slower proposer never delays the block proposal, and a faster proposer delays the block only when the expected propagation delay of a slower proposer is at least $\tau_1-\epsilon$, the attestation deadline. Note that, if the expected propagation delay of a proposer is closer to $\tau_1$, effectively, there is no competition, and the instance is the same as the existing protocol. In summary, we curtail timing games in \ourmech\ by introducing competition between two proposers to reach more attestors as early as possible, while providing the right incentives. To the best of our knowledge, this is the first work to mitigate timing games in blockchains via a game-theoretic approach.

\textbf{Contributions.}
\label{ssec:contributions}
Our key contributions are as follows:
(i) We address the timing games problem through a game-theoretic approach.
(ii) We propose \ourmech\, where two proposers are selected for each slot each proposing a block; (iii) We propose a simple yet effective method to confirm block based on the attestations; (iv) We propose reward sharing between the proposers that incentives faster block proposal; (iv) We model block proposal in \ourmech\ as a 2-player game, \emph{Latency Game}; (v) Under homogeneous proposers setting, we show that it is Nash Equilibrium for the proposers not to delay block proposal; (vi) Under heterogeneous proposers setting, we show that a faster proposer delays only when the expected time for block to reach an attestor is close to $\tau_1$.

\textbf{Organization of the paper.}
\label{ssec:paper_org}
The rest of the paper is organized as follows: Section \ref{sec:prelims} provides a brief overview of blockchain, game theory, and timing games in blockchain. Section \ref{sec:our_approach} explains the proposed approach to mitigate timing games. Section \ref{sec:thry_analysis} discusses properties of \ourmech\ and Section \ref{sec:discussion} discusses the possible implications of \ourmech. Section \ref{sec:rel_work} discusses the related work around timing games in blockchains, and Section \ref{sec:conclusion} concludes the work with directions for future work.

\section{The Model and Preliminaries}
\label{sec:prelims}
We present our blockchain framework, network model, and relevant game theory concepts to formally define the timing game problem.

\subsection{Blockchain}
\label{ssec:block_proto}

\paragraph{Blockchains.}
\label{sssec:model}
Let $\mathfrak{L}$ be the distributed database, blockchain, consisting of a sequence of blocks. The state of the blockchain at beginning of slot $\ell$ is denoted by sequence $\mathfrak{L}^\ell = (B^{\ell-1}, \ldots, B^{1}, B^{0})$, where $B^i$ is the block confirmed at slot $i$. The state changes at each slot when a new (confirmed) block gets added, \(\mathfrak{L}^{\ell+1} \leftarrow B^\ell\cdot\mathfrak{L}^\ell\), according to the underlying consensus protocol $\Xi$.  

\paragraph{Leader-based Consensus Protocol for Blockchains.}
\label{ssec:con}
We assume $\Xi$ is a leader-based consensus protocol that uses a validator selection function \(\mathscr{V}(\mathfrak{L}^{\ell},n)\) to select the committee of nodes for each slot. For any slot $\ell$, $\Xi$ selects (i) $n + 1$  nodes as \emph{validators}, \(V^{\ell}\), (ii) out of which one node is selected as the \emph{block proposer}, $P^{\ell}$ and the remaining as \emph{attestors} $A^{\ell} = \{A^\ell_0, \ldots A^\ell_{n-1}\}$. \footnote{The selection of attestors and the proposer is through some cryptographic functions such as verifiable random functions (VRF) \cite{micali1999verifiable} or distributed random beacons (DRBs)~\cite{randao}. Our description of $\Xi$ is motivated by Ethereurm. Note that many protocols, such as PBS \cite{ethereum2024a}, DFINITY's ICC protocol~\cite{camenisch2022internet}, Harmony~\cite{avarikioti2022harmony}, Snowwhite~\cite{bentov2016snow}, etc., follow the same abstraction; their underlying security primitives are different.}  
Formally, we have \(V^{\ell} \leftarrow \mathscr{V}(\mathfrak{L}^{\ell},n)\) where $A^\ell$, $P^\ell$ forms a partition on $V^\ell$. 
The proposer $P^{\ell}$ proposes the block at the beginning of the slot $\ell$ ($t=0$ in the slot timeline). $A^{\ell}$ wait till $t=\tau_1$ (\emph{attestation deadline}), to receive the block. The attestors that receive the block within $t=\tau_1$, validate and sign it, if it is valid. This signature is referred to as \emph{attestation}. $A^\ell$ broadcasts their attestations to the network, and all attestations are aggregated by the end of $\tau_2$ (\emph{aggregation deadline}). The block is confirmed if at least \thd\ attestors attest to it. $\Xi$ allows time till $\tau$ for the network to determine if the new block is confirmed and the next proposer to propose a block for the next slot $\ell+1$.  

If $P^{\ell}$'s block achieves at least \thd\ attestations, $P^\ell$ gets \emph{block reward} and attestors $A^{\ell}$ get \emph{attestation reward} if its attestation aligns with the majority and the block is confirmed. Additionally, $P^{\ell}$ collects transaction fees and can extract additional value by reordering transactions. The latter is known as \emph{maximal extractable value}, \mev~\cite{ethereum2024b}. 
Let $\mathfrak{U}$ denote the expected reward for $P^{\ell}$ over the block reward collected just before the start of the slot. That is, $\mathfrak{U}$ is the expected transaction fees and \mev\ from the block that $P^{\ell}$ can accumulate from $t=-\tau$ to $t=0$ (w.r.t. to the start of the slot). Note that each proposer's block reward is constant and does not change with delay; hence, we do not model it in the utility/rewards.

\paragraph{Objectives of Proposers}
\label{sssec:prop_attestor_objectives}
A proposer with low-latency network access can propose the block at a delay of $0 < \delta \le \tau_1$ into the slot instead of $t=0$, and capture additional transaction fees and \mev\ of $v(\delta)$. Note that $v(\delta)$ is monotonically non-decreasing in $\delta$. Thus, if its block is confirmed, its expected reward is $\mathbb{E}[\mathfrak{U}$ + $v(\delta)]$. 

Let the probability of reaching at least \thd\ attestors when the block is delayed by $\delta$ be $M_{\delta}^{\thd}$, then the objective of a proposer is to select a delay that maximizes $\mathbb{E}[\mathfrak{U}$ + $v(\delta)]$.
\begin{equation}
    \delta^{\star} =\max_{\delta \in [0,\tau_1]} \mathbb{E}[\mathfrak{U} + v(\delta)] = \max_{\delta \in [0,\tau_1]} (\mathfrak{U}+v(\delta)) \cdot M_{\delta}^{\thd}\label{eq:delta}
\end{equation}

\paragraph{Practical Parameters} The Ethereum protocol specification parameters are as follows: $\tau=12$ seconds, $\tau_1=4,\tau_2=8$, and $n=127,\thd=\lfloor\frac{2n}{3}\rfloor+1$ \cite{ethereum_validator_phase0}. The observed value of $\delta^\star$ in practice is between $2.5$ and $3$ seconds \cite{blazquez2023relays}.

Due to the delayed block proposal, either: (i) The block fails to get $\thd$ attestations with probability $1-M_\delta^{\thd}$, affecting the throughput, or (ii) The block gets confirmed and $P^{\ell}$ steals  $v(\delta)$ from $P^{\ell+1}$. We aim to modify $\Xi$ such that the maximum delay in block proposal at equilibrium $\delta^{NE}$ is as close as possible to zero. We drop $\ell$ from the notation unless required. To analyze $M_{\delta}^{\thd}$, we make certain reasonable network assumptions. Additionally, we borrow a few concepts from game theory to model proposers' strategic behavior. We elaborate on these in the following section.

\begin{table}[ht]
    \centering
    \begin{tabular*}{\linewidth}{l l}
         \hline
         \textbf{Notation} & \textbf{Intrepretation}\\
         \hline 
          &\\
         \multicolumn{2}{l}{\emph{Protocol}}\\
         $\Xi$ & A leader-based protocol\\
         $n$ & Number of attestors\\
         $\tau$  & Duration of slot\\
         $\tau_1$ & Duration of sub-slot for attestation\\
         $\thd$ & Minimum attestation required for block finalization\\
                   &\\
         \hline
                   &\\
          \multicolumn{2}{l}{\emph{Wrt. to Slot} $\ell$}\\
         $\mathfrak{L}^{\ell}$  & State of the blockchain\\
         $V^l$ & Set of validators\\
         $A^\ell $& Set of attestators, $A^\ell = \{A_0,\ldots A_{n-1}\}$\\
         $\mathbf{P}^\ell $& Set of proposers, \(\mathbf{P}^\ell= \{P_0,P_1\}\)\\
         $B^\ell_i$ & Block proposed by proposer $P^l_i$\\
         $X_i$ & \#Attestors that received $B_i$ within time $\tau_1$\\
         $Y_i$ & \#Attestors $B_i$ received as first block\\
         $\prec_j$ & Order of block reception for attestor $j$\\
         $\delta_i$ & delay in block proposal by proposer $P_i$\\
         $q_i(\delta_i)$ & Probability that $B_i$ reaches an attestor when block proposal is delayed by $\delta_i$\\
         $p_i(\delta_i,\delta_{i^+})$ & Probability that $B_i \prec B_{i^+}$  when block proposal is delayed by $\delta_i, \delta_{i^+}$ respectively\\
         $v_P(\delta)$ & $P$'s block valuation at time $\delta \in [0,\tau_1]$, \\
         $M_{\delta_i}^{\thd}$ & Probability of reaching $\ge\thd$ attestors \\
         & when block $B_i$ is delayed by $\delta_i$\\
        \(i^{+}\) & \((1-i)\)\\
         \hline
    \end{tabular*}
    \label{tab:notations}
\end{table}

\subsection{Network Model}
\label{ssec:network}
Due to network randomness, the block may reach an attestor at different times. For a proposer \(P\)
, let $f_P$ be the probability density function (PDF) with mean $\mu_P$ that represents the probability that a block sent by $P$ reaches an attestor. We assume $f_P$ to be unimodal with support \([0,\infty)\) as often considered in standard literature \cite{misic2019block,kiffer2021under,cheng2022fault, hwerbi2024delay} and observed in analysis of block propagation time in Ethereum by Kiraly and Leonardo~\cite{codexBlockDiffusion} and study on Ethereum Gossip Protocol~\cite{ethGossip2021} (More details in Appendix~\ref{assec:block_propg}) with the most of the blocks arrive between 1 and 4 seconds after the beginning of the slot. Let $q$ denote the probability that an attestor receives the block within $\tau_1$. Let $M_{\delta}^{\thd}$ be the probability that at least $\thd$ number of attestors receive the block within $\tau_1$. Then $M_{\delta}^{\thd}$ can be computed by the cumulative binomial probability, as given in Equation~\ref{eq:prdth}.

\begin{equation}
M_{\delta}^{\thd} = \sum_{k=\thd}^n \binom{n}{k}q^{k}{(1-q)}^{n-k}\\
\label{eq:prdth}
\end{equation} 

Further, we assume that the $f_P$s are reasonably peaked. To quantify this, we propose a restricted $ L^2$-norm that captures how concentrated the PDF is in a given interval -- the more ``peaked” the PDF is, the larger the $ L^2$-Norm.  
\begin{definition}[Restricted $L^2$ norm]
    A \emph{restricted $L^2$ norm} of $f_P$ on interval $[a,b]$, $L^2_{f_P}[a,b]$ is given by 
\begin{equation*}
L_{f_P}^2[a,b] = \left( \int_{a}^{b} f_P(x)^2 \, dx \right)^{1/2}\label{eqn:l2restricted}    
\end{equation*}

\end{definition}
In this work we assume that \( L^2_{f_P}[0,\tau_1] \ge \frac{1}{2\sqrt{\tau_1}}\). It is reasonable to make this assumption, as otherwise $f_p$ would have maximum mass between $[\tau_1,\infty)$, which implies most of the blocks will be missed. Importantly, our protocol works without this assumption; however, analytical guarantees about $\delta^{\star}$ become challenging.
\subsection{Game Theory}
\label{ssec:game_theory}
A strategic form game~\cite{myerson1997} is a tuple $\Gamma = \langle N,(S_i)_{i\in N},(u_i)_{i\in N}\rangle$, where $N$ is the set of players, $S_i$ is the strategy/action space for player $i$, and $u_i: \prod_{j\in N} S_j \rightarrow \mathbb{R}$ is its utility function. The most celebrated solution concept in game theory is \emph{Nash Equilibrium}, which is a strategic profile where each player’s strategy is a best response against the best response strategies of the other players.

\begin{definition}[{Nash Equilibrium}~\cite{nash1950equilibrium}]
        \label{def:ne}
            Given a game $\Gamma = \langle N,(S_i)_{i\in N},(u_i)_{i\in N}\rangle$, a strategy profile $(s_1^*, \ldots, s_n^*)$ is called a \emph{Nash equilibrium} (NE), if $\forall i \in N$, the strategy $s_i^*$ satisfies,
            \[u_i(s_i^*, s_{-i}^*) \geq u_i(s_i, s_{-i}^*), \forall s_i \in S_i,\]
            This means that unilateral deviation would not help any agent achieve higher utility.
\end{definition}

\subsection{Timing Game}
\label{sssec:strategy_space}

\emph{Timing game} in leader-based blockchain protocols is single-agent decision problem~\cite{tadelis2012game} with agent $P$ with strategy $S_P=[0,\tau_1]$, and expected utility $U_P(\delta) = \left(\mathfrak{U}+v(\delta)\right) \cdot M_{\delta}^{\thd}$. The proposer chooses a $\delta^\star$ given by Equation~\ref{eq:delta}. As attestors cannot be easily involved to mitigate the timing game, we need better policing of such strategic proposers. 
 
\section{Our Approach}
\label{sec:our_approach}
We propose \ourmech\ as a modification to $\Xi$. The key idea in \ourmech\ is that for each slot of $\Xi$, two proposers are selected who compete to win the slot. The incentives are aligned to reach attestors as fast asquickly as possible, inducing competition among proposers to publish blocks faster.
\subsection{Proposed Solution \ourmech}
We assume that \ourmech\ can securely select two proposers for each slot. For example, by invoking $\mathscr{V}(\mathfrak{L}^{\ell},n)$ of $\Xi$ twice. Note that the two-proposer selection is exogenous to \ourmech\ and as long as it is secure, our mechanism works. 
In case of \ourmech\, $\mathscr{V}(\mathfrak{L}^{\ell},n)$ outputs (i) $n+2$ nodes as validators $\mathbf{V}^\ell$ and (ii) two among these validator as the block proposers, $\mathbf{P}^\ell=\{{P}_0,{P}_1\}$. The remaining validators will be attestors $\mathbf{A}^\ell=\{A_0^\ell,\dots,A_{n-1}^\ell\}$. With two proposers, we need to update $\Xi$ for (i) Block proposition, (ii) Attestation, (iii) Winning block determination (Block finalization), and (iv) Rewards for attestors and proposers.

\textbf{Block Proposition.} In \ourmech\ both the proposers are expected to construct the blocks independently and announce their blocks at $t=0$. Let $B_0, B_1$ be the blocks proposed by $P_0, P_1$. Similar to $\Xi$, no proposer can propose more than one block in a slot. We assume that the blocks are valid for the analysis. Any such violations can be handled by \ourmech\ in the same way as $\Xi$.

\textbf{Attestation with Voting.} Note that each attestor waits till $t=\tau_1$ to receive blocks and attests to \emph{all} valid blocks received within this duration. In $\Xi$, the attestors attest to a single block in a slot. \ourmech\ allows each attestor to provide one attestation per block per proposer.\footnote{Such restrictions are in place for $\Xi$ and can be enforced similarly in \ourmech.} In addition to attestation, each attestor provides a vote to indicate the order of reception of blocks. The vote can be provided in various ways, including adding timestamps when the attestor attests to the block, or requiring the attestor to always send a vote (as additional data) along with their first attestation. Note that the presence of such votes does not incur additional voting rounds to count. Let the $\prec_j$ denote the order of reception of blocks at $A_j$. $B_i \prec_j B_{i^+}$ implies that $A_j$ received $B_i$ before $B_{i^+}$, where $i^+ = 1 -i$.  In case $B_{i^+}$ is not received at $A_j$ or in the absence of $A_j$'s attestation on $B_{i^+}$, we consider $B_i \prec_j B_{i^+}$. If an attestor fails to receive any block within $t=\tau_1$, then it makes an attestation on \emph{empty} similar to $\Xi$. If no block achieves $\thd$ attestations, the slot is considered to be missed. Note that \emph{empty} here could also be attesting to anything else or not attesting at all. For instance, in the case of Ethereum, \emph{empty} indicates that the attestor is attesting to the previous confirmed block~\cite{eth2bookUpgradingEthereum}, and the attestor skips voting in this case. We now discuss how the block is confirmed.

\textbf{Block Confirmation.} The attestations are aggregated by $\tau_2$, and at most one of the blocks is confirmed. Let \(X_i\) represent the number of attestations for \(B_i\). Similar to $\Xi$, for a block to be considered for confirmation, it has to achieve \(\thd\) attestations, i.e., $X_i \geq \thd$.

\begin{itemize}
    \item[i.] $\forall i\in\{0,1\},$ if $X_i < \thd$, then no block is confirmed for that slot
    \item[ii.] $\exists! i\in\{0,1\}$ s.t. $ X_i \ge \thd$, then only $B_i$ is confirmed.
\end{itemize} 

If both blocks obtain $\thd$ attestations and are confirmed, then the following concerns arise: (i) a proposer involved in timing games can continue to do it and get confirmed; (ii) Fair distribution of the profit from transactions common in both blocks will be challenging.

To this end, we propose selecting one of these two blocks uniformly at random (u.a.r.) using on-chain randomness. Many protocols \cite{VRaas2024, damle2024no} have used on-chain verifiable randomness, and one such technique tailored for \ourmech\ is provided in the Appendix \ref{assec:onChain_random} for reference. Note that the same mechanism used for random proposer selection from the $\Xi$ may be used, but we assume that it could be any oracle that cannot be manipulated.
$$\forall i\in\{0,1\}, \text{if } X_i \ge \thd, \text{then confirm } B_k,\text{ where } k \sim \text{Uniform}(\{0,1\})$$

Random block selection is adopted over a rule-based deterministic selection as it minimizes the attack surface for transacting entities, such as censorship of any particular transaction or block by a proposer, bribery attacks~\cite{karakostas2024blockchain}, and also ensures that honest vanilla (or non-MEV) proposers do not consistently lose to low-latency proposers (\emph{discouragement attack}~\cite{buterin2018discouragement}). If fast proposers' blocks are confirmed consistently, honest proposers who are relatively slower than fast proposers may not have their blocks confirmed, which can discourage them from skipping the block proposal. Hence, on-chain randomization prevents \ourmech\ from centralization through such attacks. 

Note that security guarantees of \ourmech\ remain similar to $\Xi$, i.e., if $\Xi$ is secure with $k$ nodes controlled by the adversary, then \ourmech\ is also secure under the same condition. The following section discusses how validator rewards in \ourmech.

\textbf{Attestor Rewards.}
The attestor rewards are similar to those of $\Xi$, where attestors receive a reward only when their attestation aligns with $\thd$ attestations and gets confirmed on the blockchain.

\textbf{Proposer Rewards}
With randomization in block finalization, a faster proposer still faces the same system as $\Xi$ with a 50\% chance. Even with randomization, the faster proposer still delays the block announcement, though earlier than $\Xi$. Reward from additional delay up to some optimal $\delta^\star$ still outweighs the decreased probability of obtaining $\thd$ attestations. 
To encourage proposers to announce their blocks as early as possible, we propose a reward-sharing policy among the proposers to encourage early block proposals.

Let \(Y_i\) represent the number of votes received for \(B_i\), i.e., the number of attestors who has $B_i \prec B_{i^+}$. A naive approach of rewarding only the proposer with a higher $Y_i$ may result in a consistent loss of rewards to honest proposers, especially vanilla proposers. This is because an honest proposer $P_i$ with $X_i \ge\thd$ gets zero reward when $Y_i < Y_{i^+}$ even if $Y_i \approx Y_{i^+}$. Furthermore, an equal split among proposers only ensures reduced rewards but does not guarantee faster block proposals. Hence, we propose a reward-sharing policy for proposers obtaining at least $\thd$ attestations $R_i \propto Y_i$ as described in Equation~\ref{eq:reward-sharing}. $R_i$ indicates the fraction of the rewards obtained from the confirmed block

\begin{equation}
\label{eq:reward-sharing}
    R_i = \begin{cases}
            \quad\ 1 & \text{if } X_i \ge \thd\ \text{and } X_{i^+} < \thd\\
            \frac{Y_i}{Y_i + Y_{i^+}}& \text{if } X_i \ge \thd\ \text{and } X_{i^+} \ge \thd\ \\\
            \quad\ 0 & \text{otherwise} 
          \end{cases}
\end{equation}

By delaying the block proposal in anticipation of a higher value block, a proposer only increases its risk of reaching attestors much later than the other proposer. Hence, the reward share from the winning block is lower. Figure~\ref{fig:slotProgEth} shows the sequence of events in \ourmech\ within a slot.

\subsection{Latency Game in \ourmech}
\label{ssec:timing_2_prop}
The competitive nature of the block proposition in \ourmech\ naturally leads to analyzing it in a game-theoretic framework. We refer to the induced game as a \emph{Latency Game}. We assume rational attestors follow the protocol, as it is in $\Xi$. With slight abuse of notation, the Latency Game is described as follows: For each slot, players are $\mathbf{P}=\{P_0, P_1\}$. Each can delay proposing their block by $\delta_i\in [0,\tau_1]$. Thus, the strategy space of $P_i$ is $[0,\tau_1]$. Let $U_i(\delta_0,\delta_1)$ be $P_i$'s expected utility when $P_0$, $P_{1}$ delay their block announcement by $\delta_0, \delta_{1}$ respectively. Thus, we define Latency Game as a tuple $\Gamma^{\mathscr{L}}=\langle\mathbf{P}, (S_i)_{i\in\mathbf{P}},(U_i(\cdot))_{i\in\mathbf{P}}\rangle$.
To determine utilities, we first analyze the block valuation.

\textbf{Block Valuation}
\label{ssec:block_valuation}
Transactions are typically announced between $[\tau_2,\tau]$ for the next slot~\cite{daian2019flash,payloadEthereumData}. To model the increase in block value over time, we analyze the timing information for blocks 21720000-21750648. For each block, we find the block value from $t\in [\tau_2,\tau]$ in the previous slot to $t\in[0,\tau_1]$ in the current slot. Figure \ref{fig:block_valuation} shows the average block growth across consecutive slots. The block value for $0\le \delta\leq\tau_1$ increases linearly, specifically, $v_P^l(s) = (1+0.03\delta)\mathfrak{U}$ for $\delta\geq0$.

\begin{figure}[ht]
    \centering
        \includegraphics[scale=0.5]{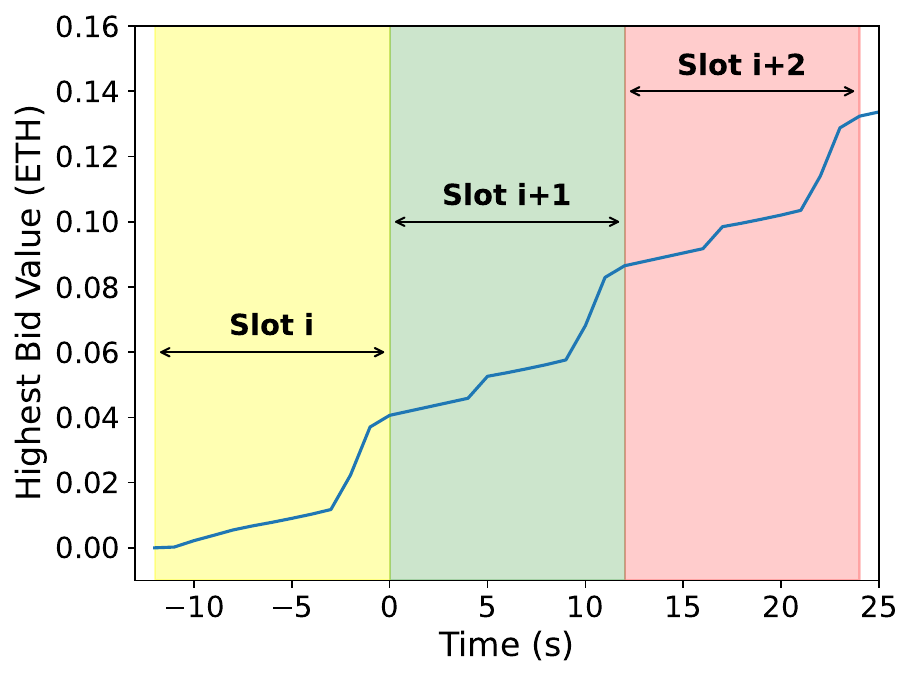}
        \caption{Average Cumulative Block Valuation across three Consecutive Slots}
    \label{fig:block_valuation}
\end{figure}

The expected reward that a proposer can generate from transactions announced in the previous slot is $\mathfrak{U}$. If a proposer delays the block by $\delta$, it can earn an additional utility of $v(\delta)$. From the data, we can safely upper-bound it as a linear function:
\begin{equation}
    v(\delta) = c\cdot\frac{\mathfrak{U}\cdot \delta}{\tau_1},\quad c\leq 1 \mbox{ and } \delta \in [0,\tau_1]
    \label{eqn:vdelta}
\end{equation}

\textbf{Proposer Utility}
The expected utility of proposer $P_i$, denoted by $U_i(\delta_0,\delta_{1})$ consists of two components. $\mathfrak{R}^1_i(\cdot)$ captures the expected block valuation when both the proposers obtain at least $\thd$ attestations. $\mathfrak{R}^2_i(\cdot)$ is the expected block valuation when only $P_i$ obtains $\thd$ attestations and $P_{i^+}$ fails. The expected utility of proposer $P_i$ is given by Equation~\ref{eq:utility_high_level}.

\begin{equation}
    U_i(\delta_i,\delta_{i^+}) = \mathfrak{R}^1_i(\delta_i,\delta_{i^+}) + \mathfrak{R}^2_i(\delta_i,\delta_{i^+})
    \label{eq:utility_high_level}
\end{equation} 

\section{Analysis of Latency Game, \texorpdfstring{$\Gamma^\mathscr{L}$}{}}
\label{sec:thry_analysis}

\subsection{Basic Results}
\label{ssec:basic_results}
The analysis of $\Gamma^{\mathscr{L}}$ relies on the following probabilities. Let $q_i(\delta_i)$ be the probability that $B_i$ reaches an attestor within $t=\tau_1$ when $P_i$ delays the block proposal by $\delta_i$. Furthermore, let $p_i(\delta_i,\delta_{i^+})$ be the probability that $B_i \prec_j B_{i^+}$ given $P_i, P_{i^+}$ delay the block proposal by $\delta_i,\delta_{i^+}$, respectively. Formally, we have  

\begin{align}
q_i(\delta_i)&= \int_0^{\tau_1 -\delta_i} f_{P_i}(x) dx\label{eq.q(delta)}\\
\hat{q}_i(\delta_i) &= 1 - q_i(\delta_i)\\
p_i(\delta_i,\delta_{i^+}) &= \int_{0}^{\tau_1-\delta_i}{f_{P_i}(x)\int_{x+\delta_i -\delta_{i^+}}^{\tau_1-\delta_{i^+}}{f_{P_{i^+}}(y)}}\,dy\,dx \label{eq.p(delta)}   
\end{align}

From Equations \ref{eq.q(delta)} and \ref{eq.p(delta)}, we state the following lemmas:

\begin{lemma}\label{lemma:1}
    $\forall i \in \{0,1\}, q_i(\delta_i)$ is monotonically decreasing with $\delta_i$ 
\end{lemma}
\begin{proof}
    \begin{align*}
    q_i(\delta_i) &= \int_{0}^{\tau_1-\delta_i}{f_{P_i}(x)}\,dx\\
    q_i'(\delta_i) &= \frac{d}{d\delta_i}n{\int_{0}^{\tau_1-\delta_i}{f_{P_i}(x)}\,dx}\\
    &=\bigg(\frac{\partial}{\partial \delta_i}{(\tau_1-\delta_i)}\bigg)f_{P_i}(\tau_1-\delta_i)\\
    &= -f_{P_i}(\tau_1-\delta_i)\le 0 \quad(\because f_{P_i}\ge0)
    \end{align*}
    \end{proof}

\begin{lemma}\label{lemma:2}
    $\forall i \in \{0,1\}, p_i(\delta_i,\delta_{i^+})$ monotonically decreases with $\delta_i$. 
\end{lemma}
\begin{proof} 
   \hspace{1.38em}$p_i = \int_{0}^{\tau_1-\delta_i}{f_{P_i}(x)\int_{x+\delta_i-\delta_{i^+}}^{\tau_1-\delta_{i^+}}{f_{P_{i^+}}(y)}}\,dy\,dx$
\begin{align*}
    \frac{d}{d\delta_i}p_i(\delta_i,\delta_{i^+}) &= \frac{d}{d\delta_i}\int_{0}^{\tau_1-\delta_i}{f_{P_i}(x)\int_{x+\delta_i-\delta_{i^+}}^{\tau_1-\delta_{i^+}}{f_{P_{i^+}}(y)}}\,dy\,dx \\
    &= \bigg(\frac{\partial}{\partial \delta_i}(\tau_1-\delta_i)\bigg)f_{P_i}(\tau_1-\delta_i)\int_{\tau_1-\delta_i}^{\tau_1-\delta_i}f_{P_{i^+}}(y)\,dy + \int_{0}^{\tau_1-\delta_i}f_{P_i}(x)\frac{\partial}{\partial \delta_i}\int_{x+\delta_i-\delta_{i^+}}^{\infty}{f_{P_{i^+}}(y)}\,dy\,dx \\
    &= -\int_{0}^{\tau_1-\delta_i}{f_{P_i}(x)f_{P_{i^+}}(x+\delta_i-\delta_{i^+})}\,dx\\
\end{align*}
\end{proof}

\begin{lemma}\label{lemma:3}
    $p_{i}(\delta_i,\delta_{i^+}) +p_{i^+}(\delta_i,\delta_{i^+}) = q_i(\delta_i) q_{i^+}(\delta_{i^+})$ 
\end{lemma}
\begin{proof}[Proof Sketch]
    Using Fubini's Theorem, we have
    \begin{align*}
       p_i(\delta_i,\delta_{i^+}) =\int_{\delta_i-\delta_{i^+}}^{\tau_1-\delta_{i^+}}{f_{P_{i^+}}(y)\int_{0}^{y-\delta_i+\delta_{i^+}}{f_{P_i}(x)}}\,dx\,dy
    \end{align*}
    \text{\quad\quad Splitting the outer integral of $p_{i^+}(\delta_i,\delta_{i^+})$, we have}
    \begin{align*}
         p_{i^+}(\delta_i,\delta_{i^+}) &=\int_{0}^{\delta_i-\delta_{i^+}}f_{P_{i^+}}(y)\int_{y-\delta_i+\delta_{i^+}}^{\tau_1-\delta_i}f_{P_i}(x)\,dx\,dy\ +\int_{\delta_i-\delta_{i^+}}^{\tau_1-\delta_{i^+}}f_{P_{i^+}}(y)\int_{y-\delta_i+\delta_{i^+}}^{\tau_1-\delta_i}f_{P_i}(x)\,dx\,dy\\
        p_i(\delta_i,\delta_{i^+}) + p_{i^+}(\delta_i,\delta_{i^+}) 
        &= \int_{\delta_i-\delta_{i^+}}^{\tau_1-\delta_{i^+}}{f_{P_{i^+}}(y)\int_{0}^{\tau_1-\delta_i}{f_{P_i}(x)}}\,dx\,dy+ \int_{0}^{\delta_i-\delta_{i^+}}f_{P_{i^+}}(y)\int_{0}^{\tau_1-\delta_i}f_{P_i}(x)\,dx\,dy\\
        &= \bigg(\int_{0}^{\tau_1}{f_{P_i}(x)}\,dx\bigg)\bigg(\int_{0}^{\tau_1}{f_{P_{i^+}}(y)}\,dy\bigg)\\
        &= q_i(\delta_i)q_{i^+}(\delta_{i^+})
    \end{align*}
\end{proof}
We defer the full proof of Lemma~\ref{lemma:3} to the Appendix~\ref{assec:main_proofs}.
 
\subsection{Utility Computation}
\label{ssec:utility_compt}
For ease of exposition in the analysis, we normalize the utilities with respect to $\mathfrak{U}$.\footnote{Multiplying utilities by $\frac{1}{\mathfrak{U}}>0$ does not change the equilibrium strategies~\cite{ijcai2024}.} Consequently, we have the following:
\begin{align}
\nonumber\mathfrak{R}^1_i &= \sum_{x,y=\thd}^{n}\sum_{w=x+y-n}^{min(x,y)}e(x,y,w)(q_{i}\hat{q}_{i^+})^{x-w} (q_{i^+}\hat{q}_{i})^{y-w} (\hat{q}_{i}\hat{q}_{i^+})^{n-(x+y-w)}\sum_{z=0}^{w}g(x,y,w,z)\left(1+ \frac{v(\delta_i)+v(\delta_{i^+})}{2}\right)\\
\nonumber &\text{where}\ e(x,y,w) = \binom{n}{w}\binom{n-w}{x-w}\binom{n-x}{y-w}, g(x,y,w,z) = \binom{w}{z}p_{i}^{z}p_{i^+}^{w-z} \frac{x-w+z}{x+y-w}\\ 
\nonumber\mathfrak{R}^2_i &= M_{\delta_i}^{\thd}(1- M_{\delta_{i^+}}^{\thd}) (1 +v(\delta_i))
\end{align}

\emph{Intuition behind $U_i(\cdot)$}
Let $x,y$ be the number of attestations received for $B_i, B_{i^+}$ respectively. Then, $\mathfrak{R}^1$ captures the expected utility when $x,y \ge \thd$ and $\mathfrak{R}^2$ captures the expected utility when $x \ge \thd\mbox{ and }y<\thd$ . 

(i)\ \underline{Case 1} When both the proposers achieve $x,y \geq \thd$ attestations, the total expected reward, due to the random block selection, is $\frac{1}{2}(1+v(\delta_i)) + \frac{1}{2}(1+v(\delta_{i^+}))$. The individual proposer's share $\frac{Y_i}{Y_i + Y_{i^+}}$ is given by $g$. Computing expected utility in this case requires finding different possible votes for $P_i$. Out of $n$ attestors, the minimum number of attestors that attest for both blocks would be $w_{min}=x+y-n$, and the maximum number of attestors that will attest for both would be $w_{max} = \min(x,y)$. Among $x$ and $y$ attestations, $x-w$ and $y-w$ attestations are exclusive to $B_i, B_{i^+}$ respectively. For each of these exclusive attestations, the blocks would have received the corresponding votes. However, among the $w$ common attesations which range from  $w_{min}$ to $w_{max}$, the votes can range from $z=0$ to $w$, depending on which block was received first. Thus, for each of the above cases, $R_i$ turns out to be $\frac{x-w+z}{x+y-w}$. So, we compute all possible $R_i$s and their respective probabilities leading to the first term.

(ii) \underline{Case 2} When $x \ge \thd,y < \thd$, $P_i$ gets the entire reward, i.e., $R_i =1$. The other terms represent the probability of this event. 

\begin{figure}[t]
    \centering
        \includegraphics[scale=0.7]{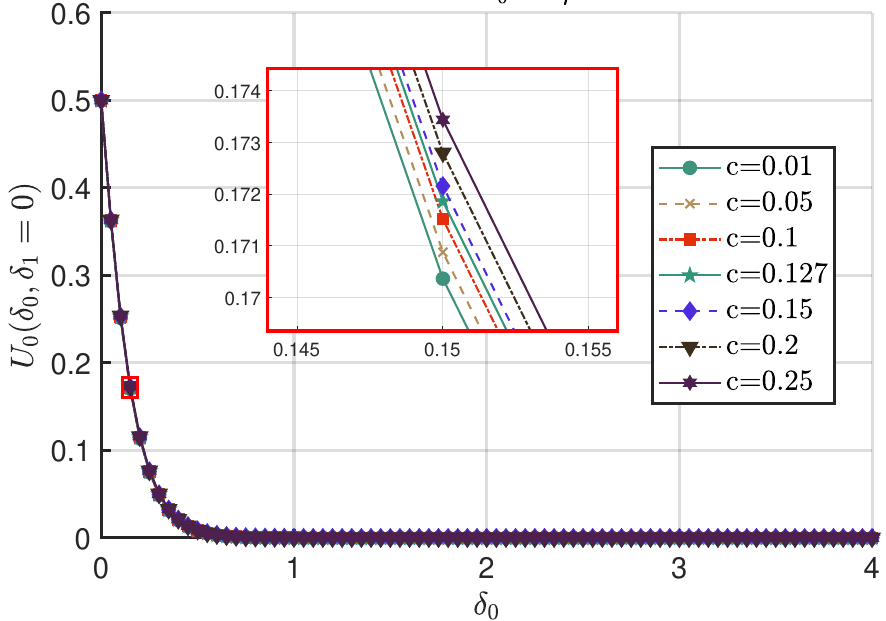}
        \caption{Proposer ($P_0$) Utility in Homogeneous Setting against an Honest Proposer ($P_1$) for $\mu =0.16$}
    \label{fig:homoMain}
\end{figure}

\subsection{Equilibrium Analysis}

To analyze \ourmech, we consider two network settings. If both the proposers for the current slot have the same network parameters, we refer to it as the \emph{homogeneous} setting. If the network parameters of proposers differ, we refer to it as the \emph{heterogeneous} setting.  

\subsubsection{Homogeneous Setting}
\label{ssec:homo_prop}
When proposers are homogeneous, we have $f_{P_0}=f_{P_1}=f$. Therefore, $\forall \delta\in[0,\tau_1], q_i(\delta) =q_{i^+}(\delta)$ and $p_i(\delta_i,\delta_{i^+}) = p_{i^+}(\delta_i,\delta_{i^+})$. 

\begin{theorem}
    \label{thm:nash}
    Under homogeneous proposers setting, with support $[0,\tau_1]$ for $f_{P_0}(f_{P_1})$, $(\delta^{NE}_{0},\delta^{NE}_{{1}}) = (0,0)$ constitutes a Nash Equilibrium of  $\hspace{0.3em}\Gamma^\mathscr{L} =\big<P,(S_i),(U_i) \big>$ in \ourmech\ is when $c\le 1$.
\end{theorem}
\begin{proof}[Proof Intuition]
    When both proposers have similar networks, the exponential decrease in the probability of reaching attestors dominates the linear increase in the block valuation $1+c\cdot \frac{\delta_i}{\tau_1}$, reducing the overall utility. Hence, it is best for the proposers not to delay.  
\end{proof}

\begin{figure}[t]
\centering
\begin{subfigure}{\linewidth}
    \centering
        \includegraphics[scale=0.4]{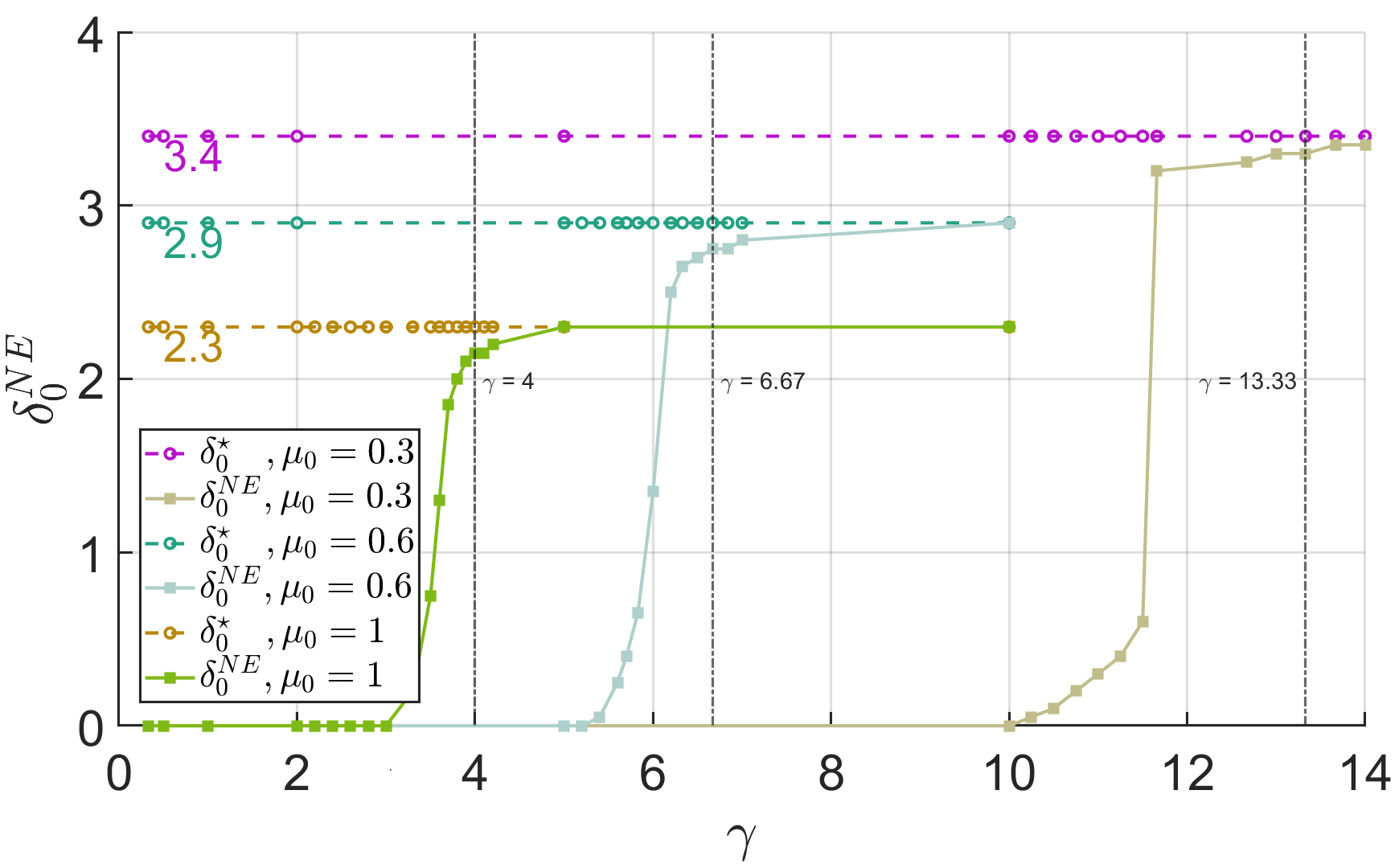}
    \caption{$\delta_0^{NE}$ vs $\gamma$ in ${D}_1$ ($P_0$ is fast)}
    \label{fig:D1}
\end{subfigure}

\begin{subfigure}{\linewidth}
    \centering
    \includegraphics[scale=0.4]{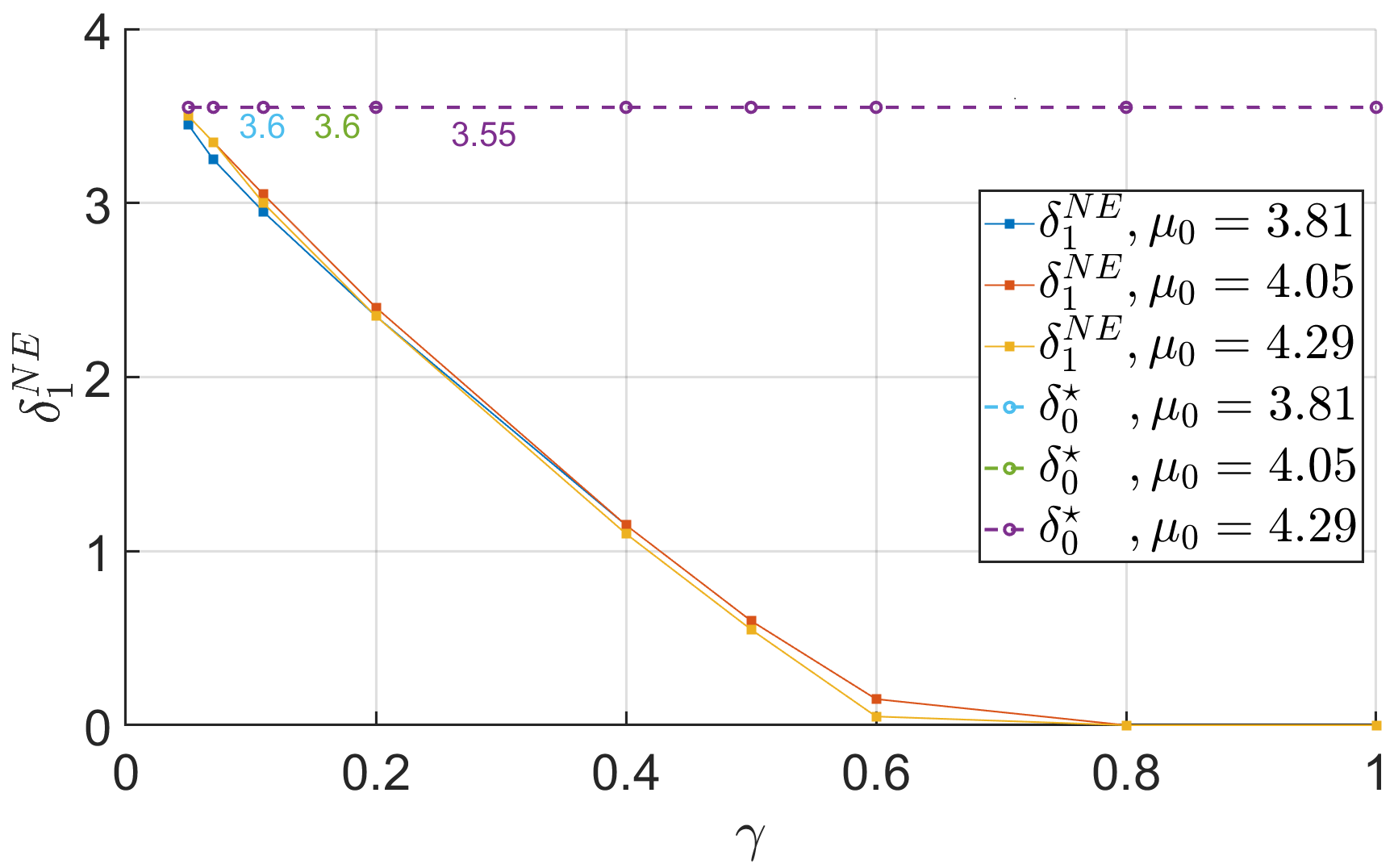}
    \caption{$\delta_1^{NE}$ vs $\gamma$ in ${D}_2$ ($P_0$ is slow)}
    \label{fig:D2}
\end{subfigure}
\caption{Equilibrium strategies for $\Gamma^{\mathscr{D}}$ for $\zeta=0.05$, $c = 1$}
\label{fig:rvsGamma}
\end{figure}

We defer the full proof of Theorem~\ref{thm:nash} to Appendix~\ref{assec:nash_equilibrium}. Figure \ref{fig:homoMain} shows the utility of the proposer in the homogeneous setting when the other proposer is honest or proposes its block on time. Observe that the utility of proposer is maximized for $\delta_0 = 0$ i.e., $U_i(\delta_i=0,\delta_{i^+}=0) > U_i(\delta_i^\prime,\delta_{i^+}=0), \forall \delta_i^\prime>0$. 

\subsubsection{Heterogeneous Setting}
\label{ssec:hetero_prop}
Computing the PSNE analytically for general network models in the heterogeneous setting is complex. Hence, we discretize the strategy space and model a Bi-matrix version of $\Gamma^{\mathscr{L}}$ as $\Gamma^{\mathscr{D}}=\langle\mathbf{P},(S_i^\zeta),(U_i)\rangle$, where $S_i^\zeta=\{0,\zeta,2\zeta,\dots,\tau_1\}$, is $\zeta-$discretized strategy space of $S_i$. We consider a small discretization step $\zeta << \frac{1}{\tau_1}$ for analysis to minimize discretization error. We assume $U_i$ is Lipschitz continuous, and therefore, $\Gamma^{\mathscr{D}}$ well approximates $\Gamma^{\mathscr{L}}$. More details about $U_i$ are provided in Appendix~\ref{assec:utility}.

As empirically observed in practice, typically, packet delays follow a Gamma distribution~\cite{gammap2p2007}. Hence, we assume $f_P$s follow the same. Let $\mu_0$ be the expected time required by $P_0$'s block to reach an attestor in the network. For analyzing $\Gamma^\mathscr{D}$, we compute the Nash Equilibrium of $\Gamma^\mathscr{D}$ under two scenarios $D_1,D_2$  with different network distributions for $P_0$, $P_1$ represented, with slight abuse of notation, by tuple $(f_{P_0},f_{P_1})$: 
\begin{itemize}
    \item ${D}_1=\{(f_{P_0}, f_{P_1}):$ $P_0$ is fast, with $\mu_0 \in \{0.075\tau_1,0.15\tau_1,0.25\tau_1\}$, and $P_1$ with $\mu_1=\gamma\mu_0, \gamma\in\mathbb{R}\}$ for carefully chosen $\gamma$ values.\
    \item ${D}_2=\{(f_{P_0},f_{P_1}) :$ $P_0$ is slow, with $\mu_0 \in \{0.95\tau_1,\tau_1,1.05\tau_1\}$ and $P_1$ with $\mu_1=\gamma\mu_0, \gamma\in\mathbb{R}\}$ for carefully chosen $\gamma$ values.
\end{itemize}

We generate utility matrices using Equation~\ref{eq:utility_high_level} for 75 games in $\Gamma^{\mathscr{D}}$ from both ${D}_1$ and ${D}_2$. We use $\tau_1=4$ as in Ethereum~\cite{ethereum_validator_phase0} and we list (i) $f_P$ parameters in detail, (ii) utilities, and (iii) $\delta_0^{NE},\delta_1^{NE}$ for these games in Appendix \ref{assec:curMechvsOurMech}. We observe that the proposers in $\Xi$ have a much higher utility by playing timing games compared to \ourmech. This reduction is due to the reward-sharing policy based on block reception.
From the analyses (Appendix \ref{assec:curMechvsOurMech}), we claim the following:

\begin{lemma}
In heterogeneous settings with distributions  $(f_{P_0},f_{P_1}) \in {D}_1 \cup {D}_2$, a proposer $P_i$ in \ourmech\ with $\mu_i > \mu_{i^+}$, the equilibrium delay  $\delta_{i}^{NE} =0$. 
\end{lemma}

\begin{lemma}
In heterogeneous settings with distributions $(f_{P_0},f_{P_1}) \in {D}_1 \cup {D}_2$, a fast proposer $P_i$ with $\mu_i < \mu_{i^+}$  and $\mu_{i^+}<<\tau_1 $, the equilibrium delay $\delta_{i}^{NE} =0$. 
\end{lemma}

Note that if $\mu_{i^+}> \tau_1$ implies it is a very slow proposer, and the probability of obtaining $\thd$ attestations will be low.  For such parameter settings, \ourmech\ is approximately the same as $\Xi$ and hence $\Gamma^{\mathscr{D}}$ reduces to the single-agent decision problem (Section \ref{sssec:strategy_space}) and hence, a timing game would happen. We can expect $\delta_i^{NE}\rightarrow\delta_i^\star$, which is shown in Figures ~\ref{fig:D1} and~\ref{fig:D2}.

\section{Further Discussion}
\label{sec:discussion}

In the above discussions, we assumed that strategic players were not colluding. In this section, we state our claims about certain colluding attacks. 

\begin{proposition}[Validator Collusion]
   \label{lemma:adv_attack}
   The probability of a successful timing attack in \ourmech\ under validator collusion is quadratically smaller than in $\Xi$. Specifically, if the attack succeeds in $\Xi$ during a slot with probability $p_{\mathrm{s}}$, the attack succeeds in \ourmech\ with probability $p_{\mathrm{s}}^2$. 
\end{proposition}
\begin{proof}
  Let $V_\mathrm{s}$ represent the set of colluding validators with combined stake $s$. Let the probability that a proposer $P \in V_s$ is selected by $\mathscr{V}$ be $p_\mathrm{s}$. Then the probability of timing games by colluding validators with stake $\mathrm{s}$ is
      $$p_{\mathrm{s}} = Pr(P\in V_\mathrm{s}, \mathscr{V})$$ 
    Similarly, the probability of timing games in \ourmech\ by colluding validators with stake $\mathrm{s}$ is 
    \begin{align*}
        Pr\left(\text{Timing games in \ourmech}\right) &= Pr(\mathbf{P} \in V_{\mathrm{s}},\mathscr{V})  \\
      &= Pr(P_0\in V_{\mathrm{s}},\mathscr{V})\times Pr(P_1 \in V_{\mathrm{s}},\mathscr{V})\\ &=p_{\mathrm{s}}^2
    \end{align*}
\end{proof}

Note that the implicit assumption of independence across the selection of proposers is because $\mathscr{V}$ in practice considers each node as a different identity.

\section{Related Work}
\label{sec:rel_work}

\textbf{Game Theory and Blockchains.}
\label{ssec:gt_blk}
The intersection of game theory and blockchain protocols has become a focal point for ensuring system stability and incentive alignment in blockchains. Daian et al.~\cite{daian2019flash} formalize Maximal Extractable Value (MEV) as profits arising from transaction reordering in decentralized exchanges and show how frontrunning and related strategies distort validator incentives and may threaten consensus stability. Following this work, Heimbach et. al.~\cite{heimbach2022} model sandwich attacks as a strategic game and propose mechanism adjustments that eliminate profitable attack equilibria. 

Other lines of work include MEV redistribution. Braga et. al.~\cite{braga2024} dynamic MEV-sharing mechanisms are proposed to allocate extractable value across participants over time, reducing concentration and improving incentive alignment. Further, another study by Tarun et. al.~\cite{chitra2022} studies MEV redistribution in Proof-of-Stake systems, analyzes how concentrating MEV rewards can weaken economic security, and proposes redistribution schemes to mitigate centralization risks. Rasheed et al.~\cite{Rasheed2025} propose a Shapley value-based framework that models revenue from private transaction matchmaking as a cooperative game and distributes earnings according to participants’ marginal contributions. Together, these works apply non-cooperative and cooperative game-theoretic tools to analyze MEV extraction, revenue allocation, and strategic manipulation in blockchain systems.

BitcoinF~\cite{BitcoinF2020} proposes a game-theoretic approach to prevent strategic mining deviations in the ``transaction fee only" model in Bitcoin. Building on the concept of equitable distribution, Jain et al.~\cite{NFairness2021} investigate network fairness and scalability, analyzing how network latencies can create unfair advantages. They propose mechanisms to prevent faster nodes from disproportionately dominating the consensus process. More recently, Damle et al.~\cite{damle2024no} expand the scope of Transaction Fee Mechanism design, introducing frameworks that achieve fairness even in the absence of traditional transaction fees, thereby reinforcing the robustness of decentralized systems against manipulative strategic players.

\textbf{Timing Games and Their Impact.}
\label{ssec:timing_impact}
\"{O}z et al.~\cite{BurakTiming2023} show that timing games are profitable for proposers and discuss how they increase the risk of missed slots. Further, Schilling et al.~\cite{schwarzschilling2023timemoneystrategictiming} show that, for any deadline enforced by the attestors, the proposers can delay until close to the deadline to achieve the threshold attestations.

\textbf{Proposer-Boost and Honest Reorgs.}
\label{ssec:boost_reorgs}
To encourage rational behavior (delaying the block proposal) toward honest behavior (proposing at the beginning of the slot), the Ethereum Foundation has proposed the Proposer Score Boost (PSB) and Honest Reorgs. Proposer Score Boost~\cite{githubLMDScoreBoosting} grants the proposer a fork-choice boost equivalent to 40\%  of the full attestation weight only for the duration of the slot. Honest Reorgs~\cite{githubHonestReorgs} allows honest proposers to use PSB during the slot to forcibly reorg (replace and remove) blocks with attestation weight below 20\%. Ever since PSB was implemented, there has been an increase in reorgs~\cite{reorgReorgpics}. Neither solution concretely addresses the problem, as the proposers can still strategically delay the block proposing to get the minimum required votes. 

\textbf{Proposer Committees.}
\label{ssec:prop_comittee}
Recent works~\cite{inclusion1,inclusion2} explored inclusion lists in which a committee of nodes provides only transaction inputs for block construction, and a single block producer constructs and publishes the block to the network. 
Subsequent works~\cite{stouka2025multiple,garimidi2025transaction} explored the design of a transaction fee mechanism (TFM) with inclusion lists and proposer committees. A transaction fee mechanism is the set of rules that determines how users pay to get their transactions included in a block, and how those fees are distributed.

All of the works mentioned above focus on preventing transaction censorship, not improving timing in blockchains. Hence, they differ from our work.

\textbf{Relay Enforcement.}
\label{ssec:relay_enforce}
In the case of Ethereum, the \emph{Proposer-Builder Separation}, shortly known as PBS \cite{ethereumEIP7732}, enabled the proposers to auction the right to build a block to entities known as \emph{builders}, which build profitable blocks and bid for inclusion by the proposer, often via third-party relays \cite{mevSaga,rasheed2024evolution}. The proposer selects the highest bid and commits to relay the block. A recent proposal~\cite{ethresearTimingGames} proposes enforcing honest block proposal timing by having relays refuse to forward late bids from builders. However, this approach is ineffective. Restrictive relays could result in the creation of new, less stringent relays trusted by builders and proposers. Relays only need the trust of builders and proposers; their alignment with broader network values is irrelevant. Thus, competitive pressures discourage relays from enforcing overly strict rules, and their capacity to regulate such timing games remains inherently limited. 

\textbf{Monitoring and Penalty for Missed Slots.}
\label{ssec:monitor}
In addition, two alternative solutions were proposed \cite{ethresearTimingGames}. The former proposes checking the consistency of each proposer's block proposal and cross-referencing it with relays to determine whether the proposer is indeed playing timing games. However, this method is only moderately effective. The latter proposes to strictly penalize proposers for missing blocks; however, this penalizes honest proposers who organically missed a slot.

\section{Conclusion }
\label{sec:conclusion}
In this work, we show that the timing games can be mitigated by inducing competition with just one additional proposer. We show the proposer's Nash equilibrium under different scenarios. In a homogeneous proposers setting, we show that not delaying the block proposal is the Nash Equilibrium. In a heterogeneous proposers setting, we show that not delaying the block proposal is the Nash Equilibrium for the faster proposer, provided the slow proposer is not extremely slow, i.e., the expected time to reach an attestation is not close to the attestation deadline. We believe our analysis provides new insights into modeling consensus protocols by accounting for timing.
\bibliographystyle{unsrt}  
\bibliography{references}
\appendix


\newpage
\section{Main Proofs}
\label{asec:proofs}

\subsection{Probability Proofs}
\label{assec:main_proofs}

\noindent\textbf{Lemma 3.} \emph{$p_{i}(\delta_i,\delta_{i^+}) +p_{i^+}(\delta_i,\delta_{i^+}) = q_i(\delta_i) q_{i^+}(\delta_{i^+})$} 

\begin{proof}

\begin{eqnarray*}
    p_i(\delta_i,\delta_{i^+}) &=& \int_{0}^{\tau_1-\delta_i}{f_{P_i}(x)\int_{x+\delta_i-\delta_{i^+}}^{\tau_{1}-\delta_{i^+}}{f_{P_{i^+}}(y)}}\,dy\,dx \\
    && 0 \le x \le \tau_1 -\delta_i\\
    && x+\delta_i-\delta_{i^+} \le y \le \tau_1-\delta_i\\
    && \implies x \le \min(\tau_1 -\delta_i, y-\delta_i+\delta_{i^+})\\
    && \delta_i -\delta_{i^+} \le y \le \tau_1 - \delta_{i^+}\\
    &=& \int_{\delta_i-\delta_{i^+}}^{\tau_1-\delta_{i^+}}{f_{P_{i^+}}(y)\int_{0}^{\min(y-\delta_i+\delta_{i^+},\tau_1-\delta_i)}{f_{P_i}(x)}}\,dx\,dy \quad (\text{By Fubini's Theorem}) \\
    &=& \int_{\delta_i-\delta_{i^+}}^{\tau_1-\delta_{i^+}}{f_{P_{i^+}}(y)\int_{0}^{y-\delta_i+\delta_{i^+}}{f_{P_i}(x)}}\,dx\,dy\\
    p_{i^+}(\delta_i,\delta_{i^+}) &=& \int_{0}^{\tau_{1}-\delta_{i^+}}{f_{P_{i^+}}(y)\int_{y+\delta_{i^+}-\delta_i}^{\tau_1-\delta_i}{f_{P_i}(x)}}\,dx\,dy \\
    &=& \int_{0}^{\delta_i-\delta_{i^+}}f_{P_{i^+}}(y)\int_{y-\delta_i+\delta_{i^+}}^{\tau_1-\delta_i}f_{P_i}(x)\,dx\,dy + \int_{\delta_i-\delta_{i^+}}^{\tau_1-\delta_{i^+}}f_{P_{i^+}}(y)\int_{y-\delta_i+\delta_{i^+}}^{\tau_1-\delta_i}f_{P_i}(x)\,dx\,dy\\
    p_i(\delta_i,\delta_{i^+}) + p_{i^+}(\delta_i,\delta_{i^+}) &=& \int_{\delta_i-\delta_{i^+}}^{\tau_1-\delta_{i^+}}{f_{P_{i^+}}(y)\int_{0}^{y-\delta_i+\delta_{i^+}}{f_{P_i}(x)}}\,dx\,dy + \int_{\delta_i-\delta_{i^+}}^{\tau_1-\delta_{i^+}}f_{P_{i^+}}(y)\int_{y-\delta_i+\delta_{i^+}}^{\tau_1-\delta_i}f_{P_i}(x)\,dx\,dy\\
    && + \int_{0}^{\delta_i-\delta_{i^+}}f_{P_{i^+}}(y)\int_{0}^{\tau_1-\delta_i}f_{P_i}(x)\,dx\,dy\\
    &=& \int_{\delta_i-\delta_{i^+}}^{\tau_1-\delta_{i^+}}{f_{P_{i^+}}(y)\int_{0}^{y-\delta_i+\delta_{i^+}}{f_{P_i}(x)}}\,dx\,dy + \int_{0}^{\delta_i-\delta_{i^+}}f_{P_{i^+}}(y)\int_{0}^{\tau_1-\delta_i}f_{P_i}(x)\,dx\,dy\\
    &=& \int_{\delta_i-\delta_{i^+}}^{\tau_1-\delta_{i^+}}{f_{P_{i^+}}(y)\int_{0}^{\tau_1-\delta_i}{f_{P_i}(x)}}\,dx\,dy + \int_{0}^{\delta_i-\delta_{i^+}}f_{P_{i^+}}(y)\int_{0}^{\tau_1-\delta_i}f_{P_i}(x)\,dx\,dy\\
    &=& \bigg(\int_{0}^{\tau_1}{f_{P_i}(x)}\,dx\bigg)\bigg(\int_{0}^{\tau_1}{f_{P_{i^+}}(y)}\,dy\bigg)\\
    &=& q_i(\delta_i)q_{i^+}(\delta_{i^+})
\end{eqnarray*}
\end{proof}

\subsection{Nash Equilibrium}
\label{assec:nash_equilibrium}
\noindent\textbf{Theorem 1.} \emph{Under homogeneous proposers setting, if support of $f_{P_0}(f_{P_1})$ is $[0,\tau_1]$, $(\delta^{NE}_{0},\delta^{NE}_{{1}}) = (0,0)$ constitutes a Nash Equilibrium of the Latency Game $\Gamma^\mathscr{L} =\big<P,(S_i),(U_i) \big>$ in \ourmech\ is  when $c\le 1$.}

\begin{proof}
For ease of exposition, let the normalized expected utility of $P_i$ at slot $\ell$ be given by:
\begin{align}
\label{eq:nexpu1}
U_{i}(\delta_i,\delta_{i^+})
&= \sum_{x,y=\thd}^{n}\sum_{w=x+y-n}^{\min(x,y)} \binom{n}{w}\binom{n-w}{x-w}\binom{n-x}{y-w}
   (q_{i}\hat{q}_{i^+})^{x-w} (q_{i^+}\hat{q}_{i})^{y-w}
   (\hat{q}_{i}\hat{q}_{i^+})^{n-(x+y-w)} \nonumber \\
&\qquad \times
   \sum_{z=0}^{w}\binom{w}{z} p_{i}^{z} p_{i^+}^{w-z}
   \frac{x-w+z}{x+y-w}
   \left(1+\frac{c(\delta_i + \delta_{i^+})}{2}\right) \nonumber \\
&\quad + \sum_{x=\thd}^{n}\binom{n}{x} q_{i}^{x}\hat{q}_{i}^{n-x}
   \left( 1- \sum_{y=\thd}^{n}\binom{n}{y}
   q_{i^+}^{y}\hat{q}_{i^+}^{n-y} \right)
   (1+c\delta_i)
\end{align}

\noindent\(
\text{Since, proposer are homogenous }f_{P_{0}} = f_{P_{1}} \text{ and thus, } q_i(\delta) = q_{i^+}(\delta). \text{ Further, } q_i(0)=q_{i^+}(0)=1\ \text{and } p_i(0)=p_{i^+}(0)=\frac{1}{2}\). Suppose \(\delta_{i^+}=0\), then for \(\forall \delta_i>0, q_i(\delta_i)<q_{i^+}(0)\) and \(p_i(\delta_i,0) < p_{i^+}(\delta_{i},0)\)

\noindent\(\forall x,y,w\) Equation~\ref{eq:nexpu1} is zero except for \(x+y-w=n\) and \(w=x\). Therefore,   

\begin{align*}
U_{i}(\delta_i,\delta_{i^+})&=\sum_{x=\thd}^{n}\binom{n}{x}(1-q_i)^{n-x}\sum_{z=0}^{x}\binom{x}{z}p_{i}^{z}p_{i^+}^{x-z} \frac{z}{n}
\left(1+ \frac{c\delta_i}{2}\right)&&&&&\\
&=\sum_{x=\thd}^{n}\binom{n}{x}(1-q_i)^{n-x}\frac{xp_i}{n}(p_i+p_{i^+})^{x-1}\left(1+ \frac{c\delta_i}{2}\right)&&&&&\\
& \text{From Prop.}~\ref{lemma:3}, \text{$p_i+p_{i^+} = q_iq_{i^+} = q_i$}&&&&&\\
&=p_i\sum_{x=\thd}^{n}\binom{n-1}{x-1}(1-q_i)^{n-x}(q_i)^{x-1}\left(1+ \frac{c\delta_i}{2}\right)&&&&&\\
&=p_i\sum_{h=\thd-1}^{n-1}\binom{n-1}{h}(1-q_i)^{(n-1)-h}(q_i)^{h}\left(1+ \frac{c\delta_i}{2}\right)&&&&&
\end{align*}

\noindent To show that \(\delta_i=0\) is best response to $\delta_{i^+}=0$. It is enough to show that \(\frac{d}{ds}U_i^{M,l}<0\)

\noindent Let \(g_i = \sum_{h=\thd-1}^{n-1}\binom{n-1}{h}(1-q_i)^{n-1-h}(q_i)^{h}\)

\noindent Then \(U^{M,l}_{i} = p_i\cdot g_i \cdot(1+\frac{c\delta_i}{2})\)

\begin{align*}
    & \frac{dU_i^{M,l}}{d\delta_i}<0&\label{eq:p3}\\
    &\implies p_ig_i\frac{c}{2}+g_i(1+\frac{c\delta_i}{2})\frac{dp_i}{d\delta_i}+p_i(1+\frac{c\delta_i}{2})\frac{dg_i}{d\delta_i}<0&\quad\\
    &\implies p_ig_i\frac{c\delta_i}{2}<(1+\frac{c\delta_i}{2})\left( g_i\int_0^{\tau_1-\delta_i}f(x)f(x+\delta_i)dx+p_i \frac{dg_i}{d\delta_i}\right)\\
\end{align*}
    
\(\because \left(1+\frac{c\delta_i}{2}\right) \ge 1, \text{ it is sufficient to show the following:}\)

\begin{align}
    p_ig_i\frac{c}{2}&<g_i\int_0^{\tau_1-\delta_i}f(x)f(x+\delta_i)dx+p_i\frac{dg_i}{dq_i}f(\tau_1-\delta_i)\\
    \frac{c}{2}&<\frac{1}{p_i}\int_0^{\tau_1-\delta_i}f(x)f(x+\delta_i)dx+\frac{1}{g_i}\frac{dg_i}{dq_i}f(\tau_1-\delta_i)\label{eq:p4}
\end{align} 

Substituting \(c \le \frac{1}{\tau_1}\) in Equation~\ref{eq:p4}

\begin{align}
    \frac{1}{2\tau_1} &<\frac{1}{p_i}\int_0^{\tau_1-\delta_i}f(x)f(x+\delta_i)dx+\frac{1}{g_i}\frac{dg_i}{dq_i}f(\tau_1-\delta_i)\label{eq:p5}
\end{align}

Both terms on the right-hand side of Equation~\ref{eq:p5} are positive and increase with $\delta_i$.

Substituting $\delta_i=0$ in Equation~\ref{eq:p5}

\begin{align}
    \frac{1}{2\tau_1} &<\frac{1}{p_i}\int_0^{\tau_1}f(x)^2dx+\frac{1}{g_i}\frac{dg_i}{dq_i}f(\tau_1-\delta_i)\quad\quad\quad\quad\quad\label{eq:p6}
\end{align}

\(\because \int_0^{\tau_1}f(x)^2 \ge \frac{1}{4\tau_1}\) and \(\frac{1}{p_i} \ge 2\), the following holds true \(\frac{1}{2\tau_1} <\frac{1}{p_i}\int_0^{\tau_1}f(x)^2dx\)

Therefore, Equation~\ref{eq:p6} holds and hence \(\frac{dU_i^{M,l}}{d\delta_i}<0\).
\end{proof}

\section{Additional Proofs}
\label{asec:add_proofs}

\subsection{On-chain Randomness}
\label{assec:onChain_random}

Let $\lambda \in \mathbb{N}_{\geq 1}$ be a security parameter. A cryptographic hash function is a one-way function defined as $\textsc{Hash}:{0,1}^* \rightarrow {0,1}^\lambda$. Such a function is said to be (i) \emph{collision-resistant} if the likelihood that two distinct inputs $x \neq y$ produce the same output is negligible, i.e., $\Pr[\textsc{Hash}(x) = \textsc{Hash}(y) \mid x \neq y] \leq \textsf{negl}(\lambda)$, and (ii) \emph{pre-image resistant} if the probability of finding any input $x$ such that $\textsc{Hash}(x)$ equals a given output is also negligible in $\lambda$. Here, $\textsf{negl}(\lambda)$ represents a negligible function in the security parameter. A typical example of such a function is SHA-256~\cite{gilbert2003security}.

If more than one block achieves $\thd$ attestations, one among them is selected using on-chain randomness by simulating an unbiased k-sided die roll. In the case of \ourmech\, the value of $k$ is 2. The roll outcome can be considered as follows:
\begin{align}
\label{eq:hash-1}\textsc{H}_f &= \textsc{H}\left(\big\|_{d\in\{0,\dots,k-1\}}\textsc{H}(root(B_d))\right)\\
\label{eq:hash-2}O(B_0,\ldots,B_{k-1}) &= {\arg\min}_{d\in\{0,\dots,k-1\}} d_H(\textsc{H}_f,\textsc{H}(root(B_d)))  
\end{align}
where $root(B_d)$ is the merkle root of block $B_d$, $H$ is a collision-resistant hash function, and $d_H$ is hamming distance function. For any two binary strings $\textsc{H}_1,\textsc{H}_2$ of $z$ bits $d_H(\textsc{H}_1,\textsc{H}_2) = \sum_{i\in [z]} \left[\textsc{H}_1^{(i)} \oplus \textsc{H}_2^{(i)}\right]$.

\begin{remark}
    Invoking $O(B_0,\ldots, B_{k-1})$ for blocks $B_0,\ldots, B_{k-1}$ is equivalent to selecting them via an unbiased k-sided die roll.
\end{remark}
\begin{proof}
    For any collision-resistant hash function $\textsc{H}:\{0,1\}^{*}\rightarrow\{0,1\}^{\lambda}$, the uniformity property implies $\textsc{H}(B_{k}) \in_{R} \{0,1\}^{\lambda}$~\cite{Robshaw2011HashOWF}. 

    \noindent From Equation~\ref{eq:hash-1} it can be inferred that $\forall d\in\{0,\dots,k\}, \textsc{H}(root(B_{d})) \in_{R} \{0,1\}^{\lambda}$ and therefore, $\textsc{H}\left(\big\|_{d\in\{0,\dots,k\}}\textsc{H}(root(B_d))\right) \in_{R} \{0,1\}^{\lambda}$.

    \noindent Since the hash invocation is for blocks that proposers commit, the outcome $O(B_1,\ldots, B_k) = \arg \min_{d\in[k]} d_H(\textsc{H}_f,\textsc{H}(root(B_d)))$ is random. We, therefore, get the equivalence by mapping this outcome to getting ``a face'' in an unbiased $k$-sided die roll.
\end{proof}

\subsection{Optimality of Broadcasting for  Transacting Entities}
\label{assec:broadcasting_optimality}
\noindent\textbf{Lemma 4.} \emph{The optimal strategy for a (rational) transaction-generating entity in \ourmech\ is to broadcast transactions to both proposers.} 
\begin{proof}
    \noindent Let $P^l_0, P^l_{1}$ be the proposers selected for slot $\ell$ with probability of reaching majority attestors $q_0$ and $q_{1}$ respectively. Let $t$ be the transaction sent to the network or the builder(s). 
    
    \noindent Let $\alpha_0, \alpha_1$ be the probability that transaction $t$ sent to the network is included by $P_0, P_1$ respectively. Then the probability of the transaction being included in the final block under \(P_0 \) and \(P_1\) is \(PoI_0 \) and \(PoI_1 \), respectively. 
    \begin{align*}
        PoI_0 = \frac{\alpha_0q_0}{2},&\quad PoI_1 = \frac{\alpha_1q_1}{2}
    \end{align*}
    \noindent Alternatively, let $\beta_0, \beta_1$ be the probability that the block sent by a builder is the highest bid for $P_0, P_1$, respectively. Then the probability of winning and therefore the probability of inclusion are 
    \begin{align*}
    PoI_0 = \frac{\beta_0q_0}{2},\quad PoI_1 = \frac{\beta_1q_1}{2}
    \end{align*}
    \noindent In both cases, the probability of inclusion by sending the transaction (or block) to both proposers is $PoI_0 + PoI_1 > PoI_{h}, \forall h \in\{0,1\}$. Thus, it is optimal to broadcast to both the proposers in \ourmech.     
\end{proof}

As in prior works~\cite{chen2022tips, wang2022transaction}, the inclusion probability depends on the inclusion strategy and the protocol parameters.

\section{Analysis}
\label{asec:empirical}

We perform our analysis using MATLAB and Python on a Windows 11 13th Gen Intel(R) Core(TM) i7-1335U processor at 1700 MHz with 16GB RAM. Utility computation in \ourmech\ is based on Optimization Toolbbox in MATLAB, while Nash equilibrium computation is based on the Python solver by \cite{knight2018nashpy}. Table~\ref{tab:params} shows the parameters selected for equilibrium analysis.

\subsection{Proposer Utilities in \ourmech\ under Homogenous Settings}
\label{assec:utilInHomo}
Figure \ref{fig:homogenous} shows the utility of $P_0$ in homogeneous settings when $\delta_1=0$. We analyze the homogeneous setting in multiple configurations where we fix the network parameters of $P_0$ and $ P_1$. We observe that the utility of proposer $P_0$ is maximum for $\delta_0=0$, when $P_1$ does not delay its block proposal, i.e., $\delta_1=0$.  These observations confirm Theorem \ref{thm:nash}.

\begin{figure*}[ht]
    \centering
    \begin{subfigure}{0.45\textwidth}
        \centering
        \includegraphics[width=\linewidth]{Figures/Results/Equal_Proposer/mean=0.16.pdf}
        \caption{$\mu_0=0.16$}
    \end{subfigure}
    \hfill
    \begin{subfigure}{0.45\textwidth}
        \centering
        \includegraphics[width=\linewidth]{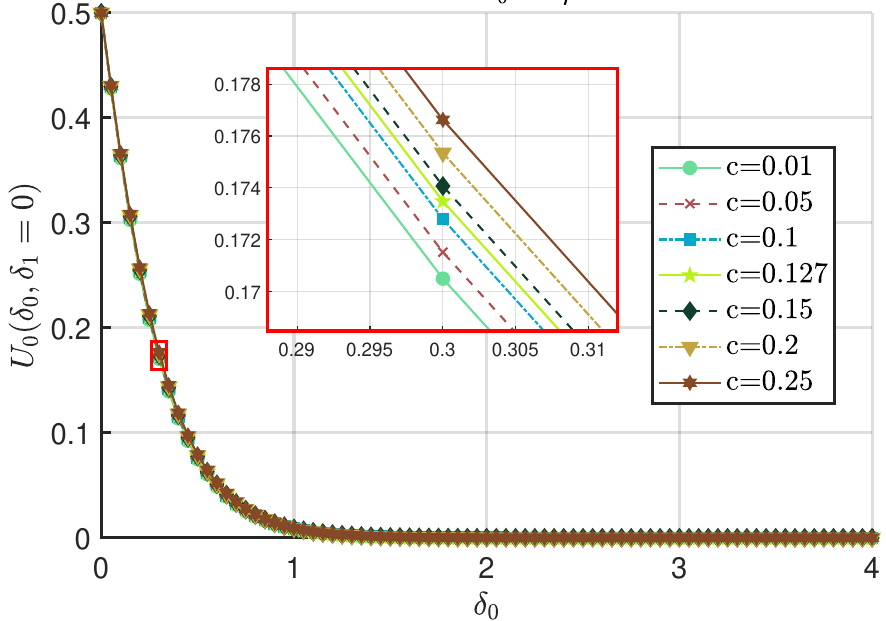}
        \caption{$\mu_0=0.33$}
        \end{subfigure}
    
    \vspace{0.5cm}
    
    \begin{subfigure}{0.45\textwidth}
        \centering
        \includegraphics[width=\linewidth]{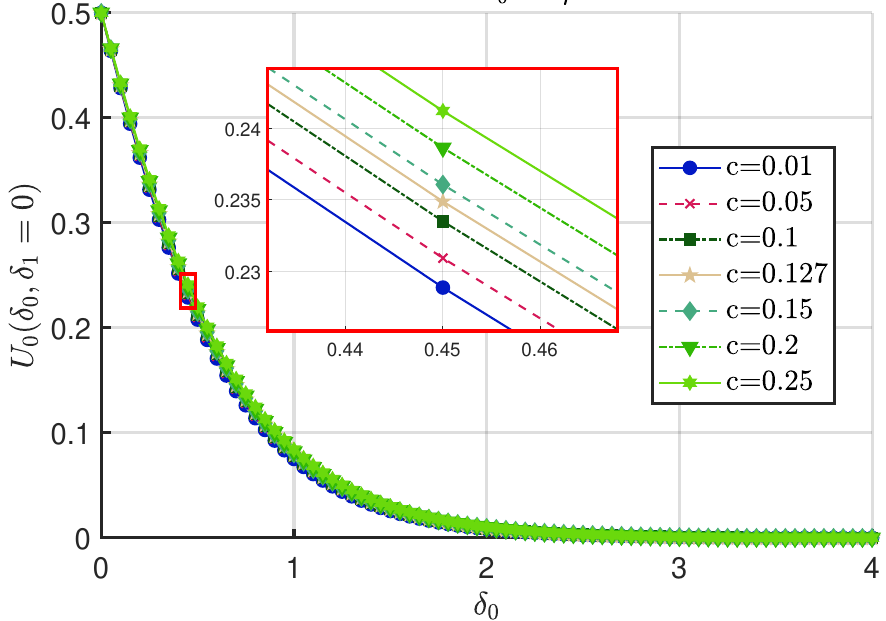}
        \caption{$\mu_0=0.66$}
    \end{subfigure}
    \hfill
    \begin{subfigure}{0.45\textwidth}
        \centering
    \includegraphics[width=\linewidth]{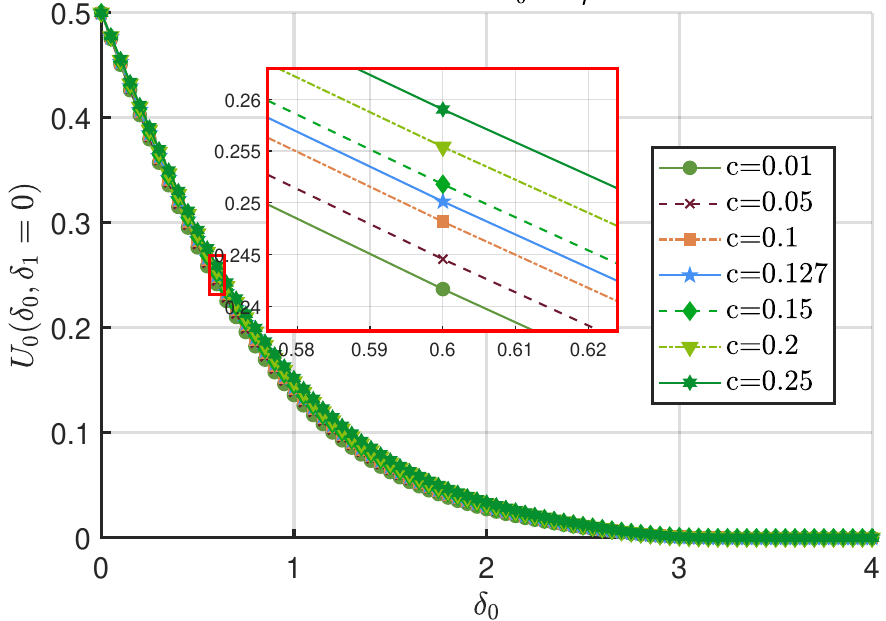}
        \caption{$\mu_0=1$}
    \end{subfigure}
    
    \caption{$U_0(\delta_0,\delta_1=0)$ vs $\delta_0$ in Homogenous Settings}
    \label{fig:homogenous}
\end{figure*}

\newpage
\subsection{Analysis \texorpdfstring{$\Xi$}{Xi} vs \ourmech\ under Heterogenous Settings}
\label{assec:curMechvsOurMech}

Table~\ref{tab:params} contains different parameters used to analyze \ourmech\ in the heterogeneous setting in two scenarios: $D_1$, where $P_0$ is fast, and $D_2$, where $P_0$ is slow. The parameters of $f_{P_0},f_{P_1}$ chosen in $D_1, D_2$ are such that the equilibrium analysis of $\Gamma^{\mathscr{D}}$ adequately captures all those cases where $P_0$ is competing against slower and faster proposers.

\begin{table}[h]
    \centering
    \caption{$f_{P_0},f_{P_1}$ parameters for Heterogenous Settings}
    \begin{tabular}{c c c c}
        \hline
         Case & $\alpha_0$ & $\lambda_0$ & $\gamma$\\
        \hline
         \multirow{3}{*}{$D_1$} & 1.5 & 5 & 0.33,0.5,1,2,5,10,10.25,10.5,10.75,11,11.25,11.5,11.66,12.67,13,13.33,13.67,14\\
          & 1.5 & 2.5 & 0.33,0.5,1,2,5,5.2,5.4,5.6,5.7,5.83,6,6.2,6.33,6.5,6.67,6.84,7,10\\
          & 2 & 2 & 0.33,0.5,1,2,2.2,2.4,2.6,2.8,3,3.3,3.5,3.6,3.7,3.8,3.9,4,4.1,4.2,5,10\\
        \hline
        \multirow{3}{*}{$D_2$} & 1.5 & 0.394 & 0.05,0.07,0.11,0.2,0.4,0.5,0.6,0.8,1,1.2,1.4\\
          & 1.5 & 0.37 & 0.05,0.07,0.11,0.2,0.4,0.5,0.6,0.8,1,1.2,1.4\\
        & 1.5 & 0.35 & 0.05,0.07,0.11,0.2,0.4,0.5,0.6,0.8,1,1.2,1.4\\
        \hline
    \end{tabular}
    \label{tab:params}
\end{table}

Figures~\ref{fig:utility_0.3},~\ref{fig:utility_0.6}, and~\ref{fig:utility_1} show the utility matrix, equilibrium delay, and utility at equilibrium (within a slot) in $\Xi$ and \ourmech\ in the case of ${D}_1$ for $\gamma$=10. Tables~\ref{tab:curVsMech1D1},~\ref{tab:curVsMech2D1}, and~\ref{tab:curVsMech3D1} show the equilibrium delay and corresponding utilities of proposer (within a slot) in $\Xi$ and \ourmech\ in the case of ${D}_1$. Tables~\ref{tab:curVsMech1D2},~\ref{tab:curVsMech2D2}, and~\ref{tab:curVsMech3D2} show the equilibrium delay and corresponding utilities of proposer (within a slot) in $\Xi$ and \ourmech\ in the case of ${D}_2$. For each scenario, the probability density function of the second proposer is fixed, and the utility of the first proposer is analyzed when it has slower, equal, and faster network connectivity.

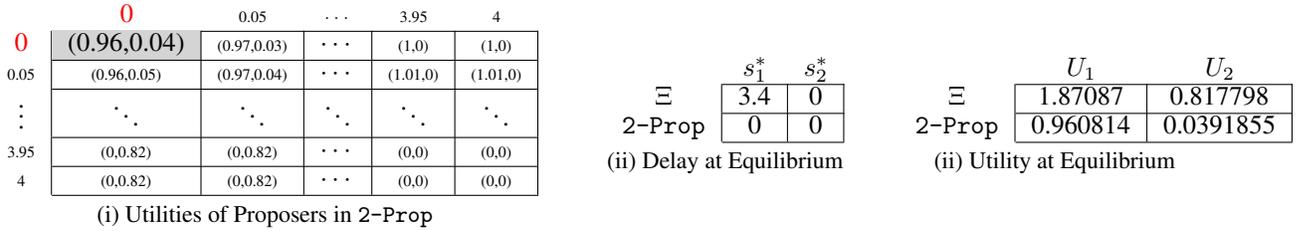
\begin{figure}[ht]
    \centering
    \input{Figures/Utilities/f1_mean=0.3/utility_f1_mean=0.3}
    \caption{Utility Matrix of Proposers in \ourmech\ with $\mu_0=0.3, \gamma=10$}
    \label{fig:utility_0.3}
\end{figure}

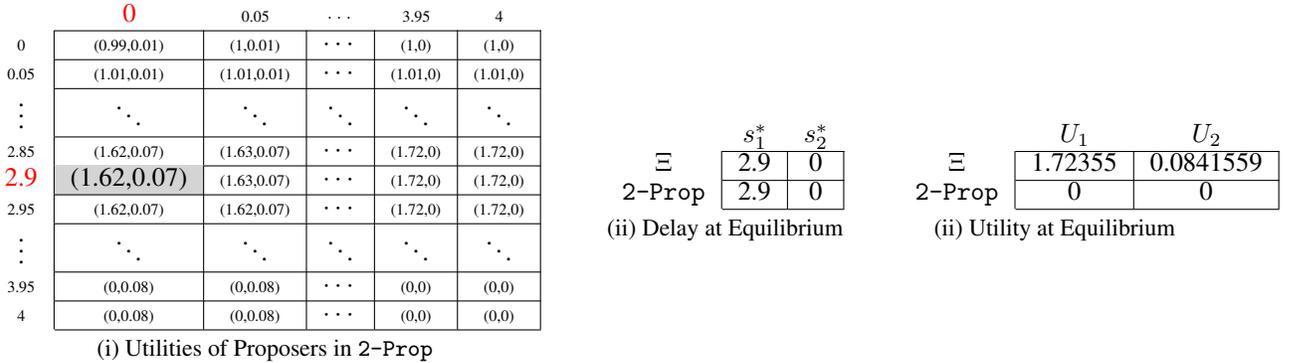
\begin{figure}[ht]
    \centering
    \input{Figures/Utilities/f1_mean=0.6/utility_f1_mean=0.6}
    \caption{Utility Matrix of Proposers in \ourmech\ with  $\mu_0=0.6, \gamma=10$}
    \label{fig:utility_0.6}    
\end{figure}

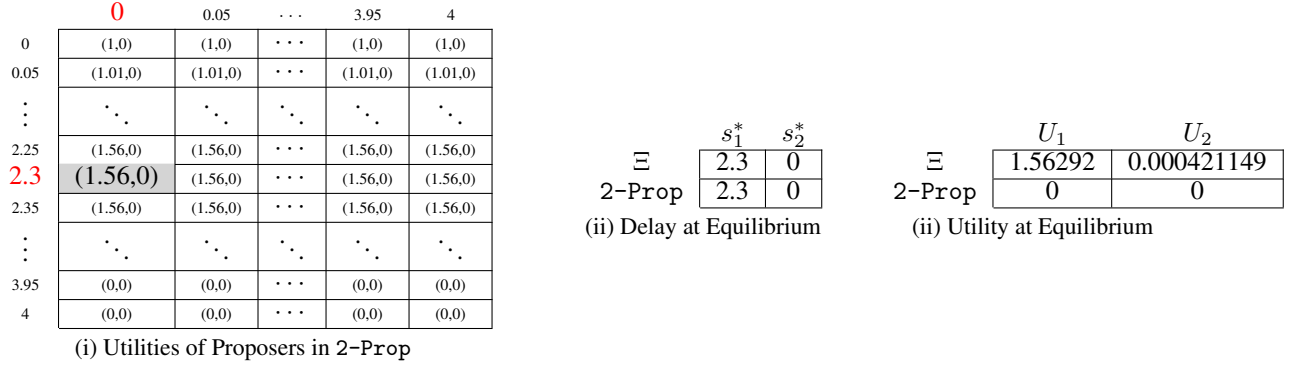
\begin{figure}[ht]
    \centering
    \input{Figures/Utilities/f1_mean=1/utility_f1_mean=1}
    \caption{Utility Matrix of Proposers in \ourmech\ with  $\mu_0=1, \gamma=10$}
    \label{fig:utility_1}
\end{figure}

\begin{table}[H]
    \centering
    \caption{\curmech\ vs \ourmech\ for $c=\frac{1}{\tau_1}=0.25, \mu_0=0.3 \text{ in } {D}_1$}
\input{Tables/D0/utility_lists_f1_mean=0.3,n=12,c=0.26,d=0.05}

    \label{tab:curVsMech1D1}
\end{table}

\begin{table}[H]
    \centering
    \caption{\curmech\ vs \ourmech\ for $c=\frac{1}{\tau_1}=0.25, \mu_0=0.6 \text{ in } {D}_1$}
    \input{Tables/D0/utility_lists_f1_mean=0.6,n=12,c=0.26,d=0.05}

    \label{tab:curVsMech2D1}
\end{table}

\begin{table}[H]
     \centering
    \caption{\curmech\ vs \ourmech\ for $c=\frac{1}{\tau_1}=0.25, \mu=1 \text{ in } {D}_1$}
     
    \input{Tables/D0/utility_lists_f1_mean=1,n=12,c=0.26,d=0.05}

    \label{tab:curVsMech3D1}
\end{table}

\begin{table}[H]
    \centering
    \caption{\curmech\ vs \ourmech\ for $c=\frac{1}{\tau_1}=0.25, \mu_0=3.81 \text{ in } {D}_2$}
    \input{Tables/D1/utility_lists_f1_mean=3.81,n=12,c=0.26,d=0.05}

    \label{tab:curVsMech1D2}
\end{table}

\begin{table}[H]
    \centering
    \caption{\curmech\ vs \ourmech\ for $c=\frac{1}{\tau_1}=0.25, \mu_0=4.05 \text{ in } {D}_2$}
    
    \input{Tables/D1/utility_lists_f1_mean=4.05,n=12,c=0.26,d=0.05}

    \label{tab:curVsMech2D2}
\end{table}

\begin{table}[H]
    \centering
    \caption{\curmech\ vs \ourmech\ for $c=\frac{1}{\tau_1}=0.25, \mu_0=4.29 \text{ in } {D}_2$}
    
    \input{Tables/D1/utility_lists_f1_mean=4.29,n=12,c=0.26,d=0.05}

    \label{tab:curVsMech3D2}
\end{table}

We observe that the proposers in $\Xi$ have a much higher utility by playing timing games compared to \ourmech. This reduction is due to the reward-sharing policy based on faster block reception. In addition, the sum of utilities in the case \ourmech\ in most cases is equal to one, indicating that no rewards are captured from the next slot.

\newpage
\section{Further Discussion}
\label{asec:further_discussion}

\subsection{Timing Games and Block Propagation}
\label{assec:block_propg}
Both builders and proposers have a mutual interest in colluding to maximize \mev\ (Maximal Extractable Value), as delaying block proposals gives them more time to analyze the transaction mempool, thereby increasing \mev\ capture. Such collusion can occur without requiring trust between the parties. Builders can continuously update their bids over time, and the proposer can select the best bid when they decide to confirm and propose their block. 
This system removes the advantage of being a large, reputable staker (since builders don't need to trust specific proposers) and levels the playing field for smaller or solo stakers. However, this also increases the likelihood that such strategies will be widely adopted, potentially undermining network fairness and throughput.
The proposer might earn the \mev\ that the next proposer could have earned by waiting longer. 
If a proposer waits too long, the block may not propagate quickly enough to the attestors, and it might not achieve a majority of votes. This leads to consensus degradation, more missed blocks, and incorrect attestations. 

The block propagation time analysis by Kiraly and Leonardo~\cite{codexBlockDiffusion} shows the probability distribution of block dissemination across different regions. This probability distribution is similar to the standard unimodal distribution assumed in other literature~\cite{misic2019block,kiffer2021under,cheng2022fault, hwerbi2024delay}. Figure \ref{fig:block_propg} shows this probability and the cumulative distribution observed for Sydney, Amsterdam, and San Francisco regions. 

Furthermore, the study by Kiffer et al.~\cite{kiffer2021under} on the gossip protocol discusses the probability of the time between subsequent receptions of a block from any peer after the first reception at their client node. This indeed can be interpreted as the probability distribution of the times at which the block information reaches the client.  Since the client learns about the block only when it receives these announcements, the distribution of these announcement times reflects the distribution of times the block reaches the client. Their analysis also shows a similar unimodal behavior.

\begin{figure}[ht]
    \centering
    \begin{subfigure}{0.6\textwidth}
        \includegraphics[width=\linewidth]{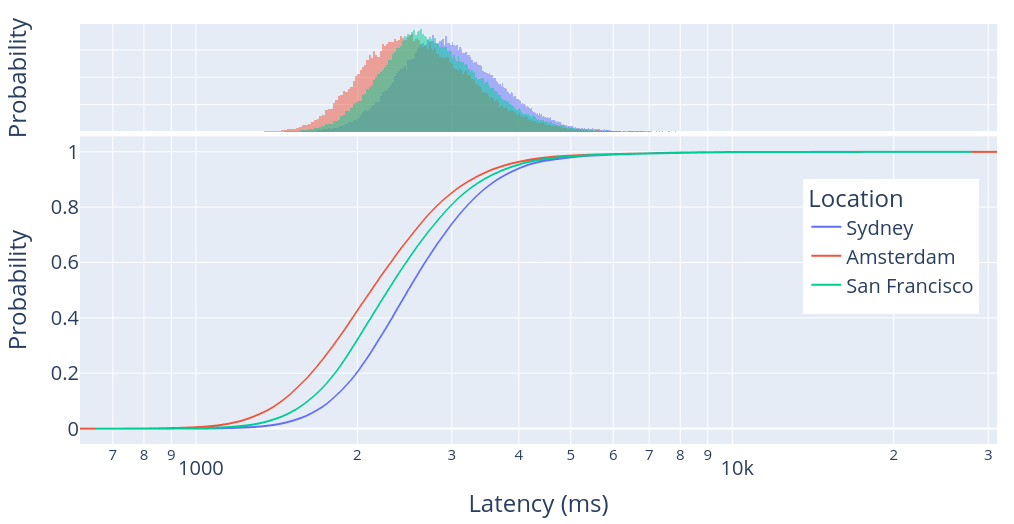}
        \caption{Block Propagation Distribution across different regions~\cite{codexBlockDiffusion}}
        \label{fig:block_propg}        
    \end{subfigure}
    \vfill
    \begin{subfigure}{0.6\textwidth}
        \includegraphics[width=\linewidth]{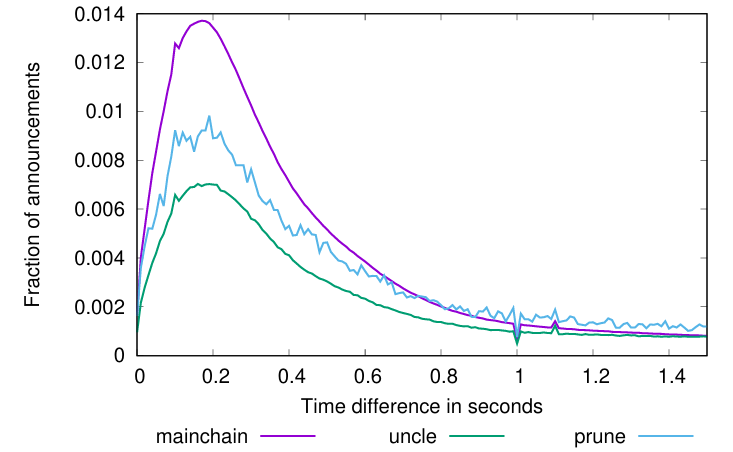}
        \caption{Distribution of Subsequent Receptions of Same Block~\cite{kiffer2021under}}
        \label{fig:block_propg_2}    
    \end{subfigure}
    \caption{Block Propagation Timing and Distribution in Ethereum}
\end{figure}

Figure~\ref{fig:delay_block_proposal} shows the percentage of blocks first received between 2024-07-13 and 2024-08-13, which are indicative of the block proposal delay. Observe that the Kiln validator proposes blocks with very high delay. Naturally, the Kiln validator turns out to be the proposer with the most missed slots (Figure~\ref{fig:missed_slots}).

\begin{figure}[ht]
    \centering
    \begin{subfigure}{0.7\textwidth}
       \includegraphics[width=\linewidth]{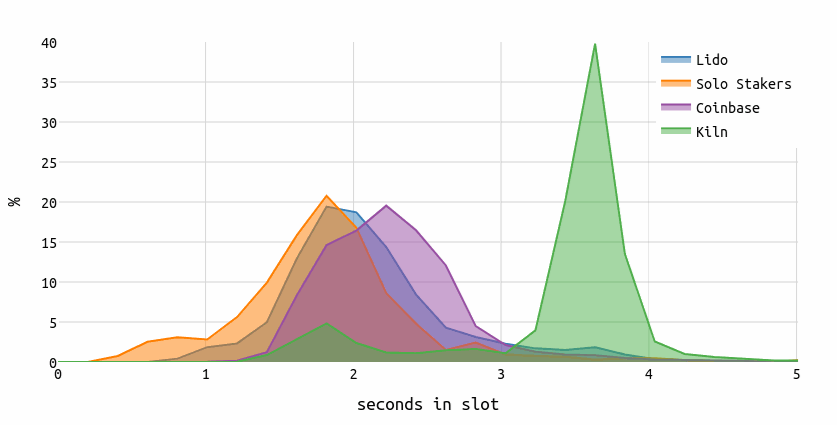}
        \caption{Percentage of Block First Seen ~\cite{ethresearAttestationsBlock}}
        \label{fig:delay_block_proposal}
    \end{subfigure}
    \vfill
    \begin{subfigure}{0.7\textwidth}
        \includegraphics[width=\linewidth]{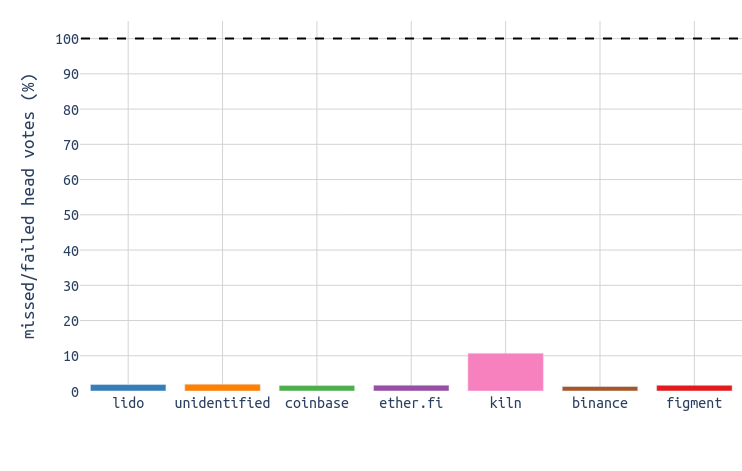}
        \caption{Percentage of Missed Slots over Proposers~\cite{ethresearAttestationsBlock}}
        \label{fig:missed_slots}
    \end{subfigure} 
    \caption{Shows the Correlation between block reception timings and Missed Slots}
\end{figure}

However, it has been observed that sophisticated fast proposers have optimized their timing strategies to minimize losses. Figure~\ref{fig:timingPerValidator} shows the fraction of slots timing games played by different proposers and the fraction of those slots they missed.

\begin{figure}[ht]
    \centering
    \includegraphics[scale=0.5]{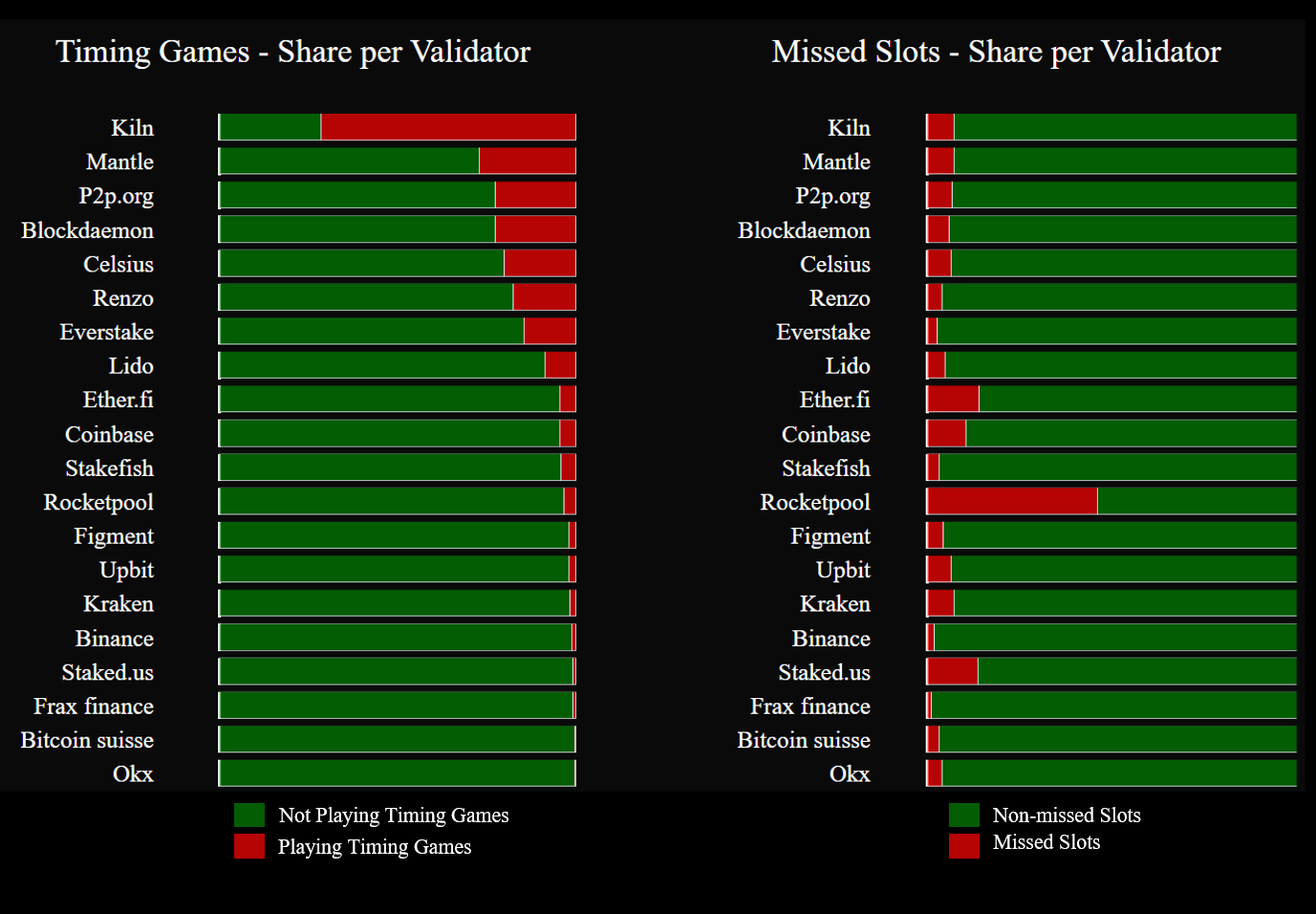}
    \caption{Timing Games Share by Different Proposers~\cite{timingTimingpics}}
    \label{fig:timingPerValidator}
\end{figure}

\newpage
\subsection{Block Valuation}
\label{assec:block_valuation}
\begin{figure}[ht]
    \centering
    \includegraphics[scale=0.6]{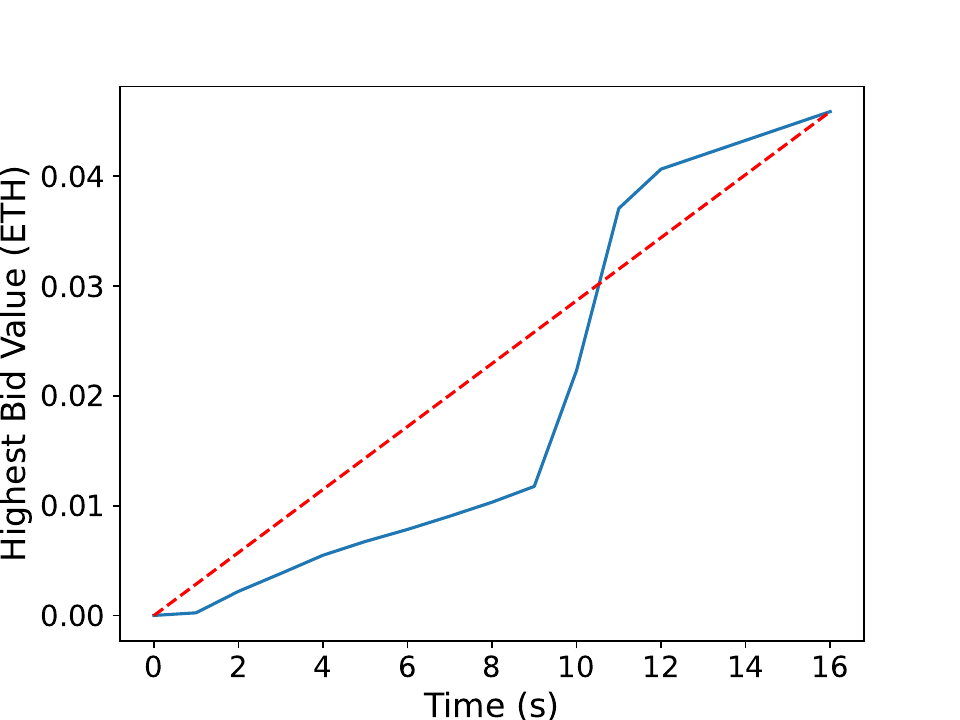}
    \caption{Average Block Valuation on Ethereum within Slot}    
\end{figure}
    
Figure \ref{fig:block_valuation_single_block} shows the variation in block valuation starting from $t\in[\tau_2,\tau]$ in the previous slot to $t\in[0,\tau_1]$ in the current slot for block 21158774 in Ethereum.

\begin{figure}[ht]
    \centering
    \includegraphics[width=\textwidth]{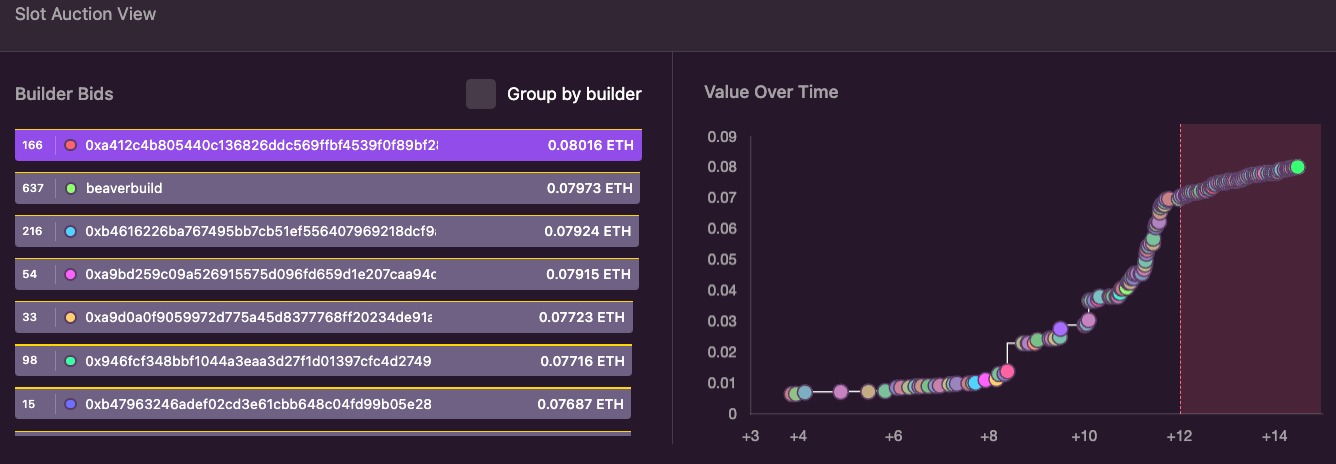}
    \caption{Block Valuation on Ethereum for block 21158774 \cite{sorellaLabs}}
    \label{fig:block_valuation_single_block} 
\end{figure}

Since an honest proposer proposes its block at $t=0$, transacting parties tend to aggressively bid to have their transaction(s) included towards the end of the previous slot. This can be observed in Ethereum Priority Gas Auctions \cite{daian2019flash} and also in current MEV-Boost auctions \cite{payloadEthereumData}. The bids from builders reduce as they shift to build the next block. Consequently, the last few bids could reach the proposer slightly delayed after the slot begins. The block value, therefore, does not increase very high relative to the block value at $t=\tau_1$ in the previous slot.

\subsection{Proposer Utility}
\label{assec:utility}

\begin{align*}
U_{i}(\delta_i,\delta_{i^+})&=\sum_{x,y=\thd}^{n}\sum_{w=x+y-n}^{min(x,y)} \binom{n}{w}\binom{n-w}{x-w}\binom{n-x}{y-w}(q_{i}\hat{q}_{i^+})^{x-w} (q_{i^+}\hat{q}_{i})^{y-w} (\hat{q}_{i}\hat{q}_{i^+})^{n-(x+y-w)}\\
&\qquad\times\sum_{z=0}^{w}\binom{w}{z}p_{i}^{z}p_{i^+}^{w-z} \frac{x-w+z}{x+y-w}\left(1+\frac{c(\delta_i +\delta_{i^+})}{2}\right)\\
&\quad+ \sum_{x=\thd}^{n}\binom{n}{x} q_{i}^{x}\hat{q}_{i}^{n-x} \left( 1- \sum_{y=\thd}^{n}\binom{n}{y} q_{i^+}^{y}\hat{q}_{i^+}^{n-y} \right) (1+c\delta_i)
\end{align*}

The expected utility, $U_i$ given above, is Lipschitz continuous, which helps the analysis of \ourmech\ in the discretized strategy space. This is because $U_i$ does not change faster than linearly and is unimodal. The unimodality stems from the fact that the maximum expected utility of $P_i$ can be either at zero delay or at some positive delay $<\tau_1$, depending on $f_{P_0}$ and $f_{P_1}$.
\newpage
\section{Code}
\label{assec:code}

\begin{listing}[ht]
\centering
\begin{lstlisting}[language=Matlab]
n=128;
T=4;
delta_res=0.05;
deltas = 0:delta_res:T;
numDeltas= length(deltas);

%%%% For Heterogeneous D1
f1_params = {[1.5,5],[1.5,2.5],[2,2]};
    f2_mean_ratios_gamma = {[0.33,0.5,1,2,5,10,10.25,10.5,10.75,11,11.25,11.5,11.66,12.67,13,13.33,13.67,14];
    [0.33,0.5,1,2,5,5.2,5.4,5.6,5.7,5.83,6,6.2,6.33,6.5,6.67,6.84,7,10];
    [0.33,0.5,1,2,2.2,2.4,2.6,2.8,3,3.3,3.5,3.6,3.7,3.8,3.9,4,4.1,4.2,5,10]};

%%%% For Heterogeneous D2
f1_params = {[1.5,0.394],[1.5,0.37],[1.5,0.35]};
f2_mean_ratios_gamma = {[0.05,0.07,0.11,0.2,0.4,0.5,0.6,0.8,1,1.2,1.4];

[0.05,0.07,0.11,0.2,0.4,0.5,0.6,0.8,1,1.2,1.4];

[0.05,0.07,0.11,0.2,0.4,0.5,0.6,0.8,1,1.2,1.4]};

f1_means = cellfun(@(x) round(x(1)/x(2),2), f1_params);

clist = [0.01,0.05,0.1,0.127,0.15,0.2,0.25];

\end{lstlisting}
\caption{Parameter initialization code}
\label{lst:parameters}
\end{listing}

\begin{listing}[ht]
\centering
\begin{lstlisting}[language=Matlab]
function compute_P(mean, s1, r1, s2, r2, T, n, deltas, delta_res,outputDir)
    %%%%%% Gamma PDFs
    f1 = @(x) gampdf(x, s1, 1/r1);
    f2 = @(x) gampdf(x, s2, 1/r2);
    ep_store = zeros(length(deltas),2);
    eq_store = zeros(length(deltas),2);
    
    q2 = integral(f2, 0, T);
    for dind = 1:1:length(deltas)
        d = deltas(dind);
        fprintf('d-value %g\n',d)
        
    %%%%%%%%% Computing p1, p2, q1, q2
        p1 = integral(@(x) f1(x) .* integral(@(y) f2(y), x+d, T, 'ArrayValued', true), 0, T-d, 'ArrayValued', true);        
        p2 = integral(@(y) f2(y) .* integral(@(x) f1(x), max(0,y-d), T-d, 'ArrayValued', true), 0, T, 'ArrayValued', true);
        q1 = integral(f1, 0, T-d);
        ep_store(dind,1) = p1;
        ep_store(dind,2) = p2;

        eq_store(dind,1) = q1;
        eq_store(dind,2) = q2; 
    end
    
    save_probabilities(ep_store,eq_store,mean,n,delta_res,outputDir);
end
\end{lstlisting}
\caption{Probability Computation code}
\label{lst:computeP}
\end{listing}

\begin{listing}[ht]
\centering
\begin{lstlisting}[language=Matlab]
function generate_proposer_instances(f2_mean_ratios_gamma, f1_params, f1_means, type)
    if strcmp(type,'heterogeneous')
        for i=1:length(f1_params)
           
            f2_mean_ratios = f2_mean_ratios_gamma{i};
            % Initialize gamma_params
            gamma_params = [];
            
            %f1_params
            s1 = f1_params{i}(1);
            r1 = f1_params{i}(2);
            mean = f1_means(i);

            fprintf('f1 mean %g\n',mean);
           
            %%% Computing Parameters
            s2 = s1;
            for r=1:length(f2_mean_ratios)
                r2 = s2/(f2_mean_ratios(r)*mean);
                gamma_params{r}=[[s1,r1];[s2,r2]];
                mean2= round(s2/r2,2);
                fprintf('mu1 %g mu2 %g ratio=%g\n',mean, mean2, round(mean2/mean,2));
            end

            save_command = sprintf('heterogeneous gamma_params f1 mean=%g.mat',mean);
            save(save_command,'gamma_params')
        end
    end
    
    if strcmp(type,'homogeneous')
        for i=1:length(f1_params)
           
            % Initialize gamma_params
            gamma_params = [];
            
            %f1_params
            s1 = f1_params{i}(1);
            r1 = f1_params{i}(2);
            mean = f1_means(i);

            fprintf('f1 mean %g\n',mean);
           
            %%% Computing Parameters
            s2 = s1;
            r2 = r1;
            
            gamma_params{r}=[[s1,r1];[s2,r2]];
            
            save_command = sprintf('homogeneous gamma_params f1 mean=%g.mat',mean);
            save(save_command,'gamma_params')
        end    
    end
    disp('Generated gamma_params');

end
\end{lstlisting}
\caption{Instance Generation code}
\label{lst:createGamma}
\end{listing}

\begin{listing}[ht]
\centering
\begin{lstlisting}[language=Matlab]
function compute_U(clist, n, mean, deltas,delta_res,outputDir)

    probmatrix= sprintf('proposer_prob_matrix f1 mean=%g, n=%d, d=%g.mat',mean, n,delta_res);
    pDirname= sprintf('%s/Probability/f1_mean=%g',outputDir,mean);
    
    pFileName = sprintf('%s/%s',pDirname,probmatrix);
    
    store = load(pFileName);
    ep_store = store.ep_store;
    eq_store = store.eq_store;
    U = zeros(length(clist),length(deltas));
    for cind = 1:length(clist) 
        c = clist(cind);
        fprintf('===> c-value %d\n', c)
        for dind = 1:length(deltas)
            d = deltas(dind);
            fprintf('d-value %g\n',d);
            p1 = ep_store(dind,1);
            p2 = ep_store(dind,2);
            q1 = eq_store(dind,1);
            q2 = eq_store(dind,2);
            
            gsum = 0;
            for k1 = floor(2*n/3):n
                g1 = 0;
                for k2 = floor(2*n/3):n
                    g2 = 0;
                    for k3 = max(0, k1 + k2 - n):min(k1, k2)
                        pv = (nchoosek(n,k3) * nchoosek(n-k3,k1-k3) * nchoosek(n-k1,k2-k3)) * ...
                             ((q1*(1-q2))^(k1-k3)) * ((q2*(1-q1))^(k2-k3)) * ...
                             ((1-q1)*(1-q2))^(n-(k1+k2-k3));
                        
                        sum1 = 0;
                        for k4 = 0:k3
                            term = nchoosek(k3, k4) * (p1^k4) * (p2^(k3-k4)) * ...
                                   ((k1-k3+k4)/(k1+k2-k3)) * (1 + (c*d)/2);
                            sum1 = sum1 + term;
                        end
                        g2 = g2 + (pv * sum1);
                    end
                    g1 = g1 + g2;
                end
                gsum = gsum + g1;
            end
            U(cind,dind) = gsum;
        end
    end
    save_utilties(U, mean,n,delta_res,outputDir);
end
\end{lstlisting}
\caption{Utility Computation code}
\label{lst:computeU}
\end{listing}

\end{document}

%% file: Figures/Utilities/f1_mean=0.3/utility_f1_mean=0.3.tex
\begin{subfigure}[h]{0.45\textwidth}
\centering
\begin{tabular}{c|c|c|c|c|c|}
\multicolumn{1}{c}{} & \multicolumn{1}{c}{\textcolor{red}{0}} & \multicolumn{1}{c}{\tiny 0.05}& \multicolumn{1}{c}{\tiny $\cdots$} & \multicolumn{1}{c}{\tiny 3.95} & \multicolumn{1}{c}{\tiny 4}\\
\cline{2-6}
{\textcolor{red}{0}}& \cellcolor{lightgray}{(0.96,0.04)} & {\tiny (0.97,0.03)} & $\cdots$& {\tiny (1,0)} & {\tiny (1,0)} \\
\cline{2-6}
{\tiny 0.05}& {\tiny (0.96,0.05)} & {\tiny (0.97,0.04)} & $\cdots$& {\tiny (1.01,0)} & {\tiny (1.01,0)} \\
\cline{2-6}
$\vdots$ & $\ddots$ & $\ddots$ & $\ddots$ & $\ddots$ & $\ddots$\\
\cline{2-6}
{\tiny 3.95}& {\tiny (0,0.82)} & {\tiny (0,0.82)} & $\cdots$& {\tiny (0,0)} & {\tiny (0,0)} \\
\cline{2-6}
{\tiny 4}& {\tiny (0,0.82)} & {\tiny (0,0.82)} & $\cdots$& {\tiny (0,0)} & {\tiny (0,0)} \\
\cline{2-6}
\end{tabular}

\smallskip

\text{\footnotesize (i) Utilities of Proposers in \ourmech}
\end{subfigure}
\hspace{0.01\textwidth}
\begin{subfigure}[h]{0.25\textwidth}
\centering
\begin{tabular}{c|c|c|}
\multicolumn{1}{c}{} & \multicolumn{1}{c}{$s_1^*$} & \multicolumn{1}{c}{$s_2^*$}\\
\cline{2-3}
\curmech & 3.4 & 0\\
\cline{2-3}
\ourmech & 0 & 0\\
\cline{2-3}
\end{tabular}

\smallskip
\text{\footnotesize (ii) Delay at Equilibrium}
\end{subfigure}\hspace{0.01\textwidth}
\begin{subfigure}[h]{0.25\textwidth}
\centering
\begin{tabular}{c|c|c|}
\multicolumn{1}{c}{} & \multicolumn{1}{c}{$U_1$} & \multicolumn{1}{c}{$U_2$}\\
\cline{2-3}
\curmech & 1.87087 & 0.817798\\
\cline{2-3}
\ourmech & 0.960814 & 0.0391855\\
\cline{2-3}
\end{tabular}

\smallskip

\text{\footnotesize (ii) Utility at Equilibrium}
\end{subfigure}

%% file: Figures/Utilities/f1_mean=0.6/utility_f1_mean=0.6.tex
\begin{subtable}[h]{0.45\textwidth}
\centering
\begin{tabular}{c|c|c|c|c|c|}
\multicolumn{1}{c}{} & \multicolumn{1}{c}{\textcolor{red}{0}} & \multicolumn{1}{c}{\tiny 0.05}& \multicolumn{1}{c}{\tiny $\cdots$} & \multicolumn{1}{c}{\tiny 3.95} & \multicolumn{1}{c}{\tiny 4}\\
\cline{2-6}
{\tiny 0}& {\tiny (0.99,0.01)} & {\tiny (1,0.01)} & $\cdots$& {\tiny (1,0)} & {\tiny (1,0)} \\
\cline{2-6}
{\tiny 0.05}& {\tiny (1.01,0.01)} & {\tiny (1.01,0.01)} & $\cdots$& {\tiny (1.01,0)} & {\tiny (1.01,0)} \\
\cline{2-6}
$\vdots$ & $\ddots$ & $\ddots$ & $\ddots$ & $\ddots$ & $\ddots$\\
\cline{2-6}
{\tiny 2.85}& {\tiny (1.62,0.07)} & {\tiny (1.63,0.07)} & $\cdots$& {\tiny (1.72,0)} & {\tiny (1.72,0)} \\
\cline{2-6}
{\textcolor{red}{2.9}}& \cellcolor{lightgray}{(1.62,0.07)} & {\tiny (1.63,0.07)} & $\cdots$& {\tiny (1.72,0)} & {\tiny (1.72,0)} \\
\cline{2-6}
{\tiny 2.95}& {\tiny (1.62,0.07)} & {\tiny (1.62,0.07)} & $\cdots$& {\tiny (1.72,0)} & {\tiny (1.72,0)} \\
\cline{2-6}
$\vdots$ & $\ddots$ & $\ddots$ & $\ddots$ & $\ddots$ & $\ddots$\\
\cline{2-6}
{\tiny 3.95}& {\tiny (0,0.08)} & {\tiny (0,0.08)} & $\cdots$& {\tiny (0,0)} & {\tiny (0,0)} \\
\cline{2-6}
{\tiny 4}& {\tiny (0,0.08)} & {\tiny (0,0.08)} & $\cdots$& {\tiny (0,0)} & {\tiny (0,0)} \\
\cline{2-6}
\end{tabular}

\smallskip

\text{\footnotesize (i) Utilities of Proposers in \ourmech}
\end{subtable}
\hspace{0.01\textwidth}
\begin{subtable}[h]{0.25\textwidth}
\centering
\begin{tabular}{c|c|c|}
\multicolumn{1}{c}{} & \multicolumn{1}{c}{$s_1^*$} & \multicolumn{1}{c}{$s_2^*$}\\
\cline{2-3}
\curmech & 2.9 & 0\\
\cline{2-3}
\ourmech & 2.9 & 0\\
\cline{2-3}
\end{tabular}

\smallskip

\text{\footnotesize (ii) Delay at Equilibrium}
\end{subtable}\hspace{0.01\textwidth}
\begin{subtable}[h]{0.25\textwidth}
\centering
\begin{tabular}{c|c|c|}
\multicolumn{1}{c}{} & \multicolumn{1}{c}{$U_1$} & \multicolumn{1}{c}{$U_2$}\\
\cline{2-3}
\curmech & 1.72355 & 0.0841559\\
\cline{2-3}
\ourmech & 0 & 0\\
\cline{2-3}
\end{tabular}

\smallskip

\text{\footnotesize (ii) Utility at Equilibrium}
\end{subtable}

%% file: Figures/Utilities/f1_mean=1/utility_f1_mean=1.tex
\begin{subtable}[h]{0.45\textwidth}
\centering
\begin{tabular}{c|c|c|c|c|c|}
\multicolumn{1}{c}{} & \multicolumn{1}{c}{\textcolor{red}{0}} & \multicolumn{1}{c}{\tiny 0.05}& \multicolumn{1}{c}{\tiny $\cdots$} & \multicolumn{1}{c}{\tiny 3.95} & \multicolumn{1}{c}{\tiny 4}\\
\cline{2-6}
{\tiny 0}& {\tiny (1,0)} & {\tiny (1,0)} & $\cdots$& {\tiny (1,0)} & {\tiny (1,0)} \\
\cline{2-6}
{\tiny 0.05}& {\tiny (1.01,0)} & {\tiny (1.01,0)} & $\cdots$& {\tiny (1.01,0)} & {\tiny (1.01,0)} \\
\cline{2-6}
$\vdots$ & $\ddots$ & $\ddots$ & $\ddots$ & $\ddots$ & $\ddots$\\
\cline{2-6}
{\tiny 2.25}& {\tiny (1.56,0)} & {\tiny (1.56,0)} & $\cdots$& {\tiny (1.56,0)} & {\tiny (1.56,0)} \\
\cline{2-6}
{\textcolor{red}{2.3}}& \cellcolor{lightgray}{(1.56,0)} & {\tiny (1.56,0)} & $\cdots$& {\tiny (1.56,0)} & {\tiny (1.56,0)} \\
\cline{2-6}
{\tiny 2.35}& {\tiny (1.56,0)} & {\tiny (1.56,0)} & $\cdots$& {\tiny (1.56,0)} & {\tiny (1.56,0)} \\
\cline{2-6}
$\vdots$ & $\ddots$ & $\ddots$ & $\ddots$ & $\ddots$ & $\ddots$\\
\cline{2-6}
{\tiny 3.95}& {\tiny (0,0)} & {\tiny (0,0)} & $\cdots$& {\tiny (0,0)} & {\tiny (0,0)} \\
\cline{2-6}
{\tiny 4}& {\tiny (0,0)} & {\tiny (0,0)} & $\cdots$& {\tiny (0,0)} & {\tiny (0,0)} \\
\cline{2-6}
\end{tabular}

\smallskip

\text{\footnotesize (i) Utilities of Proposers in \ourmech}
\end{subtable}
\hspace{0.01\textwidth}
\begin{subtable}[h]{0.25\textwidth}
\centering
\begin{tabular}{c|c|c|}
\multicolumn{1}{c}{} & \multicolumn{1}{c}{$s_1^*$} & \multicolumn{1}{c}{$s_2^*$}\\
\cline{2-3}
\curmech & 2.3 & 0\\
\cline{2-3}
\ourmech & 2.3 & 0\\
\cline{2-3}
\end{tabular}

\smallskip

\text{\footnotesize (ii) Delay at Equilibrium}
\end{subtable}\hspace{0.01\textwidth}
\begin{subtable}[h]{0.25\textwidth}
\centering
\begin{tabular}{c|c|c|}
\multicolumn{1}{c}{} & \multicolumn{1}{c}{$U_1$} & \multicolumn{1}{c}{$U_2$}\\
\cline{2-3}
\curmech & 1.56292 & 0.000421149\\
\cline{2-3}
\ourmech & 0 & 0\\
\cline{2-3}
\end{tabular}

\smallskip

\text{\footnotesize (ii) Utility at Equilibrium}
\end{subtable}

%% file: Tables/D0/utility_lists_f1_mean=0.3,n=12,c=0.26,d=0.05.tex

\begin{tabular}{lcccccccc}
\cline{1-9}
\multirow{2}{*}{$\gamma =\frac{\mu_2}{\mu_1}$} & \multicolumn{4}{c}{Equilibrium Delay} & \multicolumn{4}{c}{Equilibrium Utility}\\
\cline{2-9}
  & $\delta_0^\star$ & $\delta_0^{NE}$ & $\delta_1^\star$ & $\delta_1^{NE}$ & $u_0$ in $\Xi$ & $u_0$ in $\ourmech$ & $u_1$ in $\Xi$ & $u_1$ in $\ourmech$\\
\cline{1-9}
$\gamma=0.33$ & 3.4 & 0 & 3.75 & 0 & 1.87 & 0.2 & 1.97 & 0.8\\
\cline{1-9}
$\gamma=0.5$ & 3.4 & 0 & 3.7 & 0 & 1.87 & 0.29 & 1.95 & 0.71\\
\cline{1-9}
$\gamma=1$ & 3.4 & 0 & 3.4 & 0 & 1.87 & 0.5 & 1.87 & 0.5\\
\cline{1-9}
$\gamma=2$ & 3.4 & 0 & 2.9 & 0 & 1.87 & 0.71 & 1.72 & 0.29\\
\cline{1-9}
$\gamma=5$ & 3.4 & 0 & 1.6 & 0 & 1.87 & 0.89 & 1.34 & 0.11\\
\cline{1-9}
$\gamma=10$ & 3.4 & 0 & 0 & 0 & 1.87 & 0.96 & 0.82 & 0.04\\
\cline{1-9}
$\gamma=10.25$ & 3.4 & 0.05 & 0 & 0 & 1.87 & 0.96 & 0.79 & 0.04\\
\cline{1-9}
$\gamma=10.5$ & 3.4 & 0.1 & 0 & 0 & 1.87 & 0.97 & 0.77 & 0.05\\
\cline{1-9}
$\gamma=10.75$ & 3.4 & 0.2 & 0 & 0 & 1.87 & 0.97 & 0.74 & 0.06\\
\cline{1-9}
$\gamma=11$ & 3.4 & 0.3 & 0 & 0 & 1.87 & 0.97 & 0.71 & 0.08\\
\cline{1-9}
$\gamma=11.25$ & 3.4 & 0.4 & 0 & 0 & 1.87 & 0.98 & 0.69 & 0.09\\
\cline{1-9}
$\gamma=11.5$ & 3.4 & 0.6 & 0 & 0 & 1.87 & 0.99 & 0.66 & 0.12\\
\cline{1-9}
$\gamma=11.66$ & 3.4 & 3.2 & 0 & 0 & 1.87 & 1.11 & 0.53 & 0.51\\
\cline{1-9}
$\gamma=12.67$ & 3.4 & 3.25 & 0 & 0 & 1.87 & 1.16 & 0.5 & 0.48\\
\cline{1-9}
$\gamma=13$ & 3.4 & 3.3 & 0 & 0 & 1.87 & 1.21 & 0.46 & 0.45\\
\cline{1-9}
$\gamma=13.33$ & 3.4 & 3.3 & 0 & 0 & 1.87 & 1.26 & 0.43 & 0.42\\
\cline{1-9}
$\gamma=13.67$ & 3.4 & 3.35 & 0 & 0 & 1.87 & 1.3 & 0.4 & 0.39\\
\cline{1-9}
$\gamma=14$ & 3.4 & 3.35 & 0 & 0 & 1.87 & 1.3 & 0.4 & 0.39\\
\cline{1-9}
\end{tabular}

%% file: Tables/D0/utility_lists_f1_mean=0.6,n=12,c=0.26,d=0.05.tex

\begin{tabular}{lcccccccc}
\cline{1-9}
\multirow{2}{*}{$\gamma =\frac{\mu_2}{\mu_1}$} & \multicolumn{4}{c}{Equilibrium Delay} & \multicolumn{4}{c}{Equilibrium Utility}\\
\cline{2-9}
  & $\delta_0^\star$ & $\delta_0^{NE}$ & $\delta_1^\star$ & $\delta_1^{NE}$ & $u_0$ in  $\Xi$ & $u_0$ in $\ourmech$ & $u_1$ in  $\Xi$ & $u_1$ in $\ourmech$\\
\cline{1-9}
$\gamma=0.33$ & 2.9 & 0 & 3.6 & 0 & 1.72 & 0.2 & 1.92 & 0.8\\
\cline{1-9}
$\gamma=0.5$ & 2.9 & 0 & 3.4 & 0 & 1.72 & 0.29 & 1.87 & 0.71\\
\cline{1-9}
$\gamma=1$ & 2.9 & 0 & 2.9 & 0 & 1.72 & 0.5 & 1.72 & 0.5\\
\cline{1-9}
$\gamma=2$ & 2.9 & 0 & 2 & 0 & 1.72 & 0.71 & 1.46 & 0.29\\
\cline{1-9}
$\gamma=5$ & 2.9 & 0 & 0 & 0 & 1.72 & 0.91 & 0.82 & 0.09\\
\cline{1-9}
$\gamma=5.2$ & 2.9 & 0 & 0 & 0 & 1.72 & 0.91 & 0.78 & 0.09\\
\cline{1-9}
$\gamma=5.4$ & 2.9 & 0.05 & 0 & 0 & 1.72 & 0.92 & 0.74 & 0.09\\
\cline{1-9}
$\gamma=5.6$ & 2.9 & 0.25 & 0 & 0 & 1.72 & 0.93 & 0.69 & 0.11\\
\cline{1-9}
$\gamma=5.7$ & 2.9 & 0.4 & 0 & 0 & 1.72 & 0.94 & 0.67 & 0.13\\
\cline{1-9}
$\gamma=5.83$ & 2.9 & 0.65 & 0 & 0 & 1.72 & 0.95 & 0.64 & 0.17\\
\cline{1-9}
$\gamma=6$ & 2.9 & 1.35 & 0 & 0 & 1.72 & 0.97 & 0.6 & 0.28\\
\cline{1-9}
$\gamma=6.2$ & 2.9 & 2.5 & 0 & 0 & 1.72 & 1.02 & 0.56 & 0.45\\
\cline{1-9}
$\gamma=6.33$ & 2.9 & 2.65 & 0 & 0 & 1.72 & 1.05 & 0.53 & 0.45\\
\cline{1-9}
$\gamma=6.5$ & 2.9 & 2.7 & 0 & 0 & 1.72 & 1.1 & 0.5 & 0.43\\
\cline{1-9}
$\gamma=6.67$ & 2.9 & 2.75 & 0 & 0 & 1.72 & 1.14 & 0.46 & 0.4\\
\cline{1-9}
$\gamma=6.84$ & 2.9 & 2.75 & 0 & 0 & 1.72 & 1.19 & 0.43 & 0.37\\
\cline{1-9}
$\gamma=7$ & 2.9 & 2.8 & 0 & 0 & 1.72 & 1.22 & 0.4 & 0.35\\
\cline{1-9}
$\gamma=10$ & 2.9 & 2.9 & 0 & 0 & 1.72 & 1.62 & 0.08 & 0.07\\
\cline{1-9}
\end{tabular}

%% file: Tables/D0/utility_lists_f1_mean=1,n=12,c=0.26,d=0.05.tex

\begin{tabular}{lcccccccc}
\cline{1-9}
\multirow{2}{*}{$\gamma =\frac{\mu_2}{\mu_1}$} & \multicolumn{4}{c}{Equilibrium Delay} & \multicolumn{4}{c}{Equilibrium Utility}\\
\cline{2-9}
  & $\delta_0^\star$ & $\delta_0^{NE}$ & $\delta_1^\star$ & $\delta_1^{NE}$ & $u_0$ in  $\Xi$ & $u_0$ in  $\ourmech$ & $u_1$ in  $\Xi$ &  $u_1$ in $\ourmech$\\
\cline{1-9}
$\gamma=0.33$ & 2.3 & 0 & 3.35 & 0 & 1.56 & 0.16 & 1.86 & 0.84\\
\cline{1-9}
$\gamma=0.5$ & 2.3 & 0 & 3.1 & 0 & 1.56 & 0.26 & 1.79 & 0.74\\
\cline{1-9}
$\gamma=1$ & 2.3 & 0 & 2.3 & 0 & 1.56 & 0.5 & 1.56 & 0.5\\
\cline{1-9}
$\gamma=2$ & 2.3 & 0 & 1 & 0 & 1.56 & 0.74 & 1.17 & 0.26\\
\cline{1-9}
$\gamma=2.2$ & 2.3 & 0 & 0.75 & 0 & 1.56 & 0.77 & 1.1 & 0.23\\
\cline{1-9}
$\gamma=2.4$ & 2.3 & 0 & 0.5 & 0 & 1.56 & 0.8 & 1.03 & 0.2\\
\cline{1-9}
$\gamma=2.6$ & 2.3 & 0 & 0.3 & 0 & 1.56 & 0.82 & 0.96 & 0.18\\
\cline{1-9}
$\gamma=2.8$ & 2.3 & 0 & 0.1 & 0 & 1.56 & 0.84 & 0.9 & 0.16\\
\cline{1-9}
$\gamma=3$ & 2.3 & 0 & 0 & 0 & 1.56 & 0.86 & 0.83 & 0.14\\
\cline{1-9}
$\gamma=3.3$ & 2.3 & 0.2 & 0 & 0 & 1.56 & 0.9 & 0.72 & 0.14\\
\cline{1-9}
$\gamma=3.5$ & 2.3 & 0.75 & 0 & 0 & 1.56 & 0.93 & 0.63 & 0.2\\
\cline{1-9}
$\gamma=3.6$ & 2.3 & 1.3 & 0 & 0 & 1.56 & 0.96 & 0.59 & 0.28\\
\cline{1-9}
$\gamma=3.7$ & 2.3 & 1.85 & 0 & 0 & 1.56 & 1 & 0.54 & 0.35\\
\cline{1-9}
$\gamma=3.8$ & 2.3 & 2 & 0 & 0 & 1.56 & 1.04 & 0.5 & 0.35\\
\cline{1-9}
$\gamma=3.9$ & 2.3 & 2.1 & 0 & 0 & 1.56 & 1.08 & 0.46 & 0.33\\
\cline{1-9}
$\gamma=4$ & 2.3 & 2.15 & 0 & 0 & 1.56 & 1.12 & 0.42 & 0.31\\
\cline{1-9}
$\gamma=4.1$ & 2.3 & 2.15 & 0 & 0 & 1.56 & 1.16 & 0.38 & 0.28\\
\cline{1-9}
$\gamma=4.2$ & 2.3 & 2.2 & 0 & 0 & 1.56 & 1.2 & 0.35 & 0.26\\
\cline{1-9}
$\gamma=5$ & 2.3 & 2.3 & 0 & 0 & 1.56 & 1.41 & 0.15 & 0.11\\
\cline{1-9}
$\gamma=10$ & 2.3 & 2.3 & 0 & 0 & 1.56 & 1.56 & 0 & 0\\
\cline{1-9}
\end{tabular}

%% file: Tables/D1/utility_lists_f1_mean=3.81,n=12,c=0.26,d=0.05.tex

\begin{tabular}{lcccccccc}
\cline{1-9}
\multirow{2}{*}{$\gamma =\frac{\mu_2}{\mu_1}$} & \multicolumn{4}{c}{Equilibrium Delay} & \multicolumn{4}{c}{Equilibrium Utility}\\
\cline{2-9}
  & $\delta_0^\star$ & $\delta_0^{NE}$ & $\delta_1^\star$ & $\delta_1^{NE}$ & $u_0$ in  $\Xi$ & $u_0$ in  $\ourmech$ & $u_1$ in  $\Xi$ & $u_1$ in $\ourmech$\\
\cline{1-9}
$\gamma=0.05$ & 0 & 0 & 3.6 & 3.45 & 0.53 & 0.53 & 1.93 & 1.13\\
\cline{1-9}
$\gamma=0.07$ & 0 & 0 & 3.45 & 3.25 & 0.53 & 0.51 & 1.89 & 1.11\\
\cline{1-9}
$\gamma=0.11$ & 0 & 0 & 3.2 & 2.95 & 0.53 & 0.48 & 1.81 & 1.09\\
\cline{1-9}
$\gamma=0.2$ & 0 & 0 & 2.65 & 2.35 & 0.53 & 0.42 & 1.65 & 1.03\\
\cline{1-9}
$\gamma=0.4$ & 0 & 0 & 1.55 & 1.15 & 0.53 & 0.31 & 1.33 & 0.9\\
\cline{1-9}
$\gamma=0.5$ & 0 & 0 & 1.1 & 0.6 & 0.53 & 0.26 & 1.18 & 0.84\\
\cline{1-9}
$\gamma=0.6$ & 0 & 0 & 0.65 & 0.15 & 0.53 & 0.24 & 1.05 & 0.77\\
\cline{1-9}
$\gamma=0.8$ & 0 & 0 & 0 & 0 & 0.53 & 0.3 & 0.8 & 0.61\\
\cline{1-9}
$\gamma=1$ & 0 & 0 & 0 & 0 & 0.53 & 0.39 & 0.53 & 0.39\\
\cline{1-9}
$\gamma=1.2$ & 0 & 0 & 0 & 0 & 0.53 & 0.45 & 0.3 & 0.21\\
\cline{1-9}
$\gamma=1.4$ & 0 & 0 & 0 & 0 & 0.53 & 0.49 & 0.15 & 0.11\\
\cline{1-9}
\end{tabular}

%% file: Tables/D1/utility_lists_f1_mean=4.05,n=12,c=0.26,d=0.05.tex

\begin{tabular}{lcccccccc}
\cline{1-9}
\multirow{2}{*}{$\gamma =\frac{\mu_2}{\mu_1}$} & \multicolumn{4}{c}{Equilibrium Delay} & \multicolumn{4}{c}{Equilibrium Utility}\\
\cline{2-9}
  & $\delta_0^\star$ & $\delta_0^{NE}$ & $\delta_1^\star$ & $\delta_1^{NE}$ & $u_0$ in $\Xi$ & $u_0$ in $\ourmech$ & $u_1$ in $\Xi$ & $u_1$ in $\ourmech$\\
\cline{1-9}
$\gamma=0.05$ & 0 & 0 & 3.6 & 3.5 & 0.44 & 0.45 & 1.92 & 1.26\\
\cline{1-9}
$\gamma=0.07$ & 0 & 0 & 3.45 & 3.35 & 0.44 & 0.44 & 1.88 & 1.24\\
\cline{1-9}
$\gamma=0.11$ & 0 & 0 & 3.15 & 3.05 & 0.44 & 0.41 & 1.8 & 1.2\\
\cline{1-9}
$\gamma=0.2$ & 0 & 0 & 2.55 & 2.4 & 0.44 & 0.36 & 1.63 & 1.12\\
\cline{1-9}
$\gamma=0.4$ & 0 & 0 & 1.45 & 1.15 & 0.44 & 0.26 & 1.29 & 0.95\\
\cline{1-9}
$\gamma=0.5$ & 0 & 0 & 0.95 & 0.6 & 0.44 & 0.23 & 1.14 & 0.87\\
\cline{1-9}
$\gamma=0.6$ & 0 & 0 & 0.5 & 0.15 & 0.44 & 0.21 & 1 & 0.79\\
\cline{1-9}
$\gamma=0.8$ & 0 & 0 & 0 & 0 & 0.44 & 0.27 & 0.74 & 0.58\\
\cline{1-9}
$\gamma=1$ & 0 & 0 & 0 & 0 & 0.44 & 0.35 & 0.45 & 0.35\\
\cline{1-9}
$\gamma=1.2$ & 0 & 0 & 0 & 0 & 0.44 & 0.39 & 0.23 & 0.18\\
\cline{1-9}
$\gamma=1.4$ & 0 & 0 & 0 & 0 & 0.44 & 0.42 & 0.11 & 0.09\\
\cline{1-9}
\end{tabular}

%% file: Tables/D1/utility_lists_f1_mean=4.29,n=12,c=0.26,d=0.05.tex

\begin{tabular}{lcccccccc}
\cline{1-9}
\multirow{2}{*}{$\gamma =\frac{\mu_2}{\mu_1}$} & \multicolumn{4}{c}{Equilibrium Delay} & \multicolumn{4}{c}{Equilibrium Utility}\\
\cline{2-9}
  & $\delta_0^\star$ & $\delta_0^{NE}$ & $\delta_1^\star$ & $\delta_1^{NE}$ &  $u_0$ in $\Xi$ & $u_0$ in $\ourmech$ & $u_1$ in $\Xi$ & $u_1$ in $\ourmech$\\
\cline{1-9}
$\gamma=0.05$ & 0 & 0 & 3.55 & 3.5 & 0.37 & 0.38 & 1.92 & 1.36\\
\cline{1-9}
$\gamma=0.07$ & 0 & 0 & 3.4 & 3.35 & 0.37 & 0.37 & 1.87 & 1.34\\
\cline{1-9}
$\gamma=0.11$ & 0 & 0 & 3.1 & 3 & 0.37 & 0.34 & 1.78 & 1.29\\
\cline{1-9}
$\gamma=0.2$ & 0 & 0 & 2.5 & 2.35 & 0.37 & 0.29 & 1.61 & 1.19\\
\cline{1-9}
$\gamma=0.4$ & 0 & 0 & 1.3 & 1.1 & 0.37 & 0.22 & 1.25 & 0.98\\
\cline{1-9}
$\gamma=0.5$ & 0 & 0 & 0.8 & 0.55 & 0.37 & 0.19 & 1.1 & 0.89\\
\cline{1-9}
$\gamma=0.6$ & 0 & 0 & 0.35 & 0.05 & 0.37 & 0.17 & 0.95 & 0.79\\
\cline{1-9}
$\gamma=0.8$ & 0 & 0 & 0 & 0 & 0.37 & 0.24 & 0.67 & 0.55\\
\cline{1-9}
$\gamma=1$ & 0 & 0 & 0 & 0 & 0.37 & 0.3 & 0.37 & 0.3\\
\cline{1-9}
$\gamma=1.2$ & 0 & 0 & 0 & 0 & 0.37 & 0.34 & 0.18 & 0.15\\
\cline{1-9}
$\gamma=1.4$ & 0 & 0 & 0 & 0 & 0.37 & 0.36 & 0.08 & 0.07\\
\cline{1-9}
\end{tabular}

%% file: references.bib
@misc{nakamoto2008bitcoin,
  author       = {Satoshi Nakamoto},
  title        = {Bitcoin: A Peer-to-Peer Electronic Cash System},
  year         = {2008},
  howpublished = {\url{https://bitcoin.org/bitcoin.pdf}},
  note         = {Accessed: 2025-09-20}
}

@article{buterin2013ethereum,
  title={Ethereum white paper},
  author={Buterin, Vitalik and others},
  journal={GitHub repository},
  volume={1},
  pages={22--23},
  year={2013}
}

@misc{githubLMDScoreBoosting,
	author = {Ethereum},
	title = {Proposer LMD Score Boosting},
	howpublished = {\url{https://github.com/ethereum/consensus-specs/pull/2730}},
	year = {2021},
	note = {[Accessed 04-12-2024]},
}

@misc{githubHonestReorgs,
	author = {Ethereum},
	title = {Allow honest validators to reorg late blocks},
	howpublished = {\url{https://github.com/ethereum/consensus-specs/pull/3034}},
	year = {2023},
	note = {[Accessed 04-12-2024]},
}

@misc{ethereum2024a,
  author = {{Ethereum}},
  title = {Why is it okay to centralize block building?},
  howpublished = {\url{https://ethereum.org/en/roadmap/pbs/}},
  year = {2024},
  note = {Accessed: 2024-12-03}
}

@misc{ethereum2024b,
	author = {Ethereum},
	title = {{M}aximal extractable value ({M}{E}{V}) | ethereum.org --- ethereum.org},
	howpublished = {\url{https://ethereum.org/en/developers/docs/mev/}},
	year = {2024},
	note = {[Accessed 08-12-2024]},
}

@INPROCEEDINGS{daian2019flash,
  author={Daian, Philip and Goldfeder, Steven and Kell, Tyler and Li, Yunqi and Zhao, Xueyuan and Bentov, Iddo and Breidenbach, Lorenz and Juels, Ari},
  booktitle={2020 IEEE Symposium on Security and Privacy (SP)}, 
  title={Flash Boys 2.0: Frontrunning in Decentralized Exchanges, Miner Extractable Value, and Consensus Instability}, 
  year={2020},
  volume={},
  number={},
  address={Online},
  publisher={IEEE},
  pages={910-927},
  doi={10.1109/SP40000.2020.00040}}

@misc{eth2bookUpgradingEthereum,
	author = {Ben Edgington},
	title = {{U}pgrading {E}thereum | 3.7.1 {P}hase 0 {F}ork {C}hoice --- eth2book.info},
	howpublished = {\url{https://eth2book.info/capella/part3/forkchoice/phase0/}},
	year = {2023},
	note = {[Accessed 09-12-2024]},
}

@misc{payloadEthereumData,
	author = {Payload.de},
	title = {{E}thereum {B}lock {V}alue {A}nalytics --- payload.de},
	howpublished = {\url{https://payload.de/data/}},
	year = {2025},
	note = {[Accessed 12-01-2025]},
}

@misc{reorgReorgpics,
	author = {Toni Wahrstätter},
	title = {{R}eorg.pics --- reorg.pics},
	howpublished = {\url{https://reorg.pics/}},
	year = {2025},
	note = {[Accessed 16-01-2025]},
}

@misc{timingTimingpics,
	author = {Toni Wahrstätter},
	title = {{T}iming.pics --- timing.pics},
	howpublished = {\url{https://timing.pics/}},
	year = {2025},
	note = {[Accessed 16-01-2025]},
}

@inproceedings{kiffer2021under,
author = {Kiffer, Lucianna and Salman, Asad and Levin, Dave and Mislove, Alan and Nita-Rotaru, Cristina},
title = {Under the Hood of the Ethereum Gossip Protocol},
year = {2021},
isbn = {978-3-662-64330-3},
publisher = {Springer-Verlag},
address = {Berlin, Heidelberg},
url = {https://doi.org/10.1007/978-3-662-64331-0_23},
doi = {10.1007/978-3-662-64331-0_23},
booktitle = {Financial Cryptography and Data Security: 25th International Conference, FC 2021, Virtual Event, March 1–5, 2021, Revised Selected Papers, Part II},
pages = {437–456},
numpages = {20}
}

@article{cheng2022fault,
  title={Fault Detection Method for Wi-Fi-Based Smart Home Devices},
  author={Cheng, Kefei and Xu, Jiashun and Zhang, Liang and Xu, ChengXin and Cui, Xiaotong},
  journal={Wireless communications and mobile computing},
  volume={2022},
  number={1},
  pages={4328307},
  year={2022},
  publisher={Wiley Online Library}
}

@INPROCEEDINGS{hwerbi2024delay,
  author={Hwerbi, Khouloud and Amdouni, Ichrak and Adjih, Cédric and Jacquet, Philippe and Saidane, Leila Azouz and Laouiti, Anis},
  booktitle={2024 20th International Conference on Wireless and Mobile Computing, Networking and Communications (WiMob)}, 
  title={Delay Analysis of the BFT Blockchain Data Dissemination: Case of Narwhal Protocol}, 
  year={2024},
  volume={},
  number={},
  address={Paris, France},
  publisher={IEEE},
  pages={651-656},
  doi={10.1109/WiMob61911.2024.10770300}}

@INPROCEEDINGS{misic2019block,
  author={Misic, Jelena and Misic, Vojislav B. and Chang, Xiaolin and Motlagh, Saeideh G. and Ali, M. Zulfiker},
  booktitle={ICC 2019 - 2019 IEEE International Conference on Communications (ICC)}, 
  title={Block Delivery Time in Bitcoin Distribution Network}, 
  year={2019},
  publisher={IEEE},
  volume={},
  number={},
  address={Online},
  pages={1-7},
  doi={10.1109/ICC.2019.8761420}}

@misc{ethresearAttestationsBlock,
	author = {Ethresearcha},
	title = {{O}n {A}ttestations, {B}lock {P}ropagation, and {T}iming {G}ames --- ethresear.ch},
	howpublished = {\url{https://ethresear.ch/t/on-attestations-block-propagation-and-timing-games/20272}},
	year = {2024},
	note = {[Accessed 25-09-2025]},
}

@misc{schwarzschilling2023timemoneystrategictiming,
      title={Time is Money: Strategic Timing Games in Proof-of-Stake Protocols}, 
      author={Caspar Schwarz-Schilling and Fahad Saleh and Thomas Thiery and Jennifer Pan and Nihar Shah and Barnabé Monnot},
      year={2023},
      eprint={2305.09032},
      archivePrefix={arXiv},
      primaryClass={cs.GT},
      url={https://arxiv.org/abs/2305.09032}, 
}

@misc{ethresearTimingGames,
	author = {Ethresearchb},
	title = {{T}iming {G}ames: {I}mplications and {P}ossible {M}itigations --- ethresear.ch},
	howpublished = {\url{https://ethresear.ch/t/timing-games-implications-and-possible-mitigations/17612}},
	year = {2023},
	note = {[Accessed 09-02-2025]},
}

@inproceedings{yin2019hotstuff,
author = {Yin, Maofan and Malkhi, Dahlia and Reiter, Michael K. and Gueta, Guy Golan and Abraham, Ittai},
title = {HotStuff: BFT Consensus with Linearity and Responsiveness},
year = {2019},
isbn = {9781450362177},
publisher = {Association for Computing Machinery},
address = {New York, NY, USA},
url = {https://doi.org/10.1145/3293611.3331591},
doi = {10.1145/3293611.3331591},
booktitle = {Proceedings of the 2019 ACM Symposium on Principles of Distributed Computing},
pages = {347–356},
numpages = {10},
location = {Toronto ON, Canada},
series = {PODC '19}
}

@misc{wahrstatter2023timebribemeasuringblock,
      title={Time to Bribe: Measuring Block Construction Market}, 
      author={Anton Wahrstätter and Liyi Zhou and Kaihua Qin and Davor Svetinovic and Arthur Gervais},
      year={2023},
      eprint={2305.16468},
      archivePrefix={arXiv},
      primaryClass={cs.NI},
      url={https://arxiv.org/abs/2305.16468}, 
}

@article{chen2022tips,
  title={TIPS: Transaction inclusion protocol with signaling in DAG-based blockchain},
  author={Chen, Canhui and Chen, Xu and Fang, Zhixuan},
  journal={IEEE Journal on Selected Areas in Communications},
  volume={40},
  number={12},
  pages={3685--3701},
  year={2022},
  publisher={IEEE}
}

@article{wang2022transaction,
  title={Transaction pricing mechanism design and assessment for blockchain},
  author={Wang, Zhilin and Hu, Qin and Wang, Yawei and Xiao, Yinhao},
  journal={High-Confidence Computing},
  volume={2},
  number={1},
  pages={100044},
  year={2022},
  publisher={Elsevier}
}

@inproceedings{damle2024no,
author = {Damle, Sankarshan and Srivastava, Varul and Gujar, Sujit},
title = {No Transaction Fees? No Problem! Achieving Fairness in Transaction Fee Mechanism Design},
year = {2024},
isbn = {9798400704864},
booktitle={Proceedings of the 25th International Conference on Autonomous Agents and Multiagent Systems},
publisher = {International Foundation for Autonomous Agents and Multiagent Systems},
address = {Richland, SC},
pages = {2228–2230},
numpages = {3},
location = {Auckland, New Zealand},
series = {AAMAS '24}
}

@inproceedings{BurakTiming2023,
author = {\"{O}z, Burak and Kraner, Benjamin and Vallarano, Nicol\`{o} and Kruger, Bingle Stegmann and Matthes, Florian and Tessone, Claudio Juan},
title = {Time Moves Faster When There is Nothing You Anticipate: The Role of Time in MEV Rewards},
year = {2023},
isbn = {9798400702617},
publisher = {Association for Computing Machinery},
address = {New York, NY, USA},
url = {https://doi.org/10.1145/3605768.3623563},
doi = {10.1145/3605768.3623563},
booktitle = {Proceedings of the 2023 Workshop on Decentralized Finance and Security},
pages = {1–8},
numpages = {8},
location = {Copenhagen, Denmark},
series = {DeFi '23}
}

@Inbook{Robshaw2011HashOWF,
author="Robshaw, Matthew J. B.",
title="One-Way Function",
bookTitle="Encyclopedia of Cryptography and Security",
year="2011",
publisher="Springer US",
address="Boston, MA",
pages="887--888",
isbn="978-1-4419-5906-5",
doi="10.1007/978-1-4419-5906-5_467"
}

@misc{ethereumEIP7732,
	author = "Francesco, D'Amato and Barnabé, Monnot and Michael, Neuder and Potuz and Terence, Tsao",
	title = {{E}{I}{P}-7732: {E}nshrined {P}roposer-{B}uilder {S}eparation --- eips.ethereum.org},
	howpublished = {\url{https://eips.ethereum.org/EIPS/eip-7732}},
	year = "2024",
	note = {[Accessed 09-07-2025]},
}

@InProceedings{mevSaga,
author="Ramos, Simona
and Ellul, Joshua",
editor="Ruiz, Marcela
and Soffer, Pnina",
title="The MEV Saga: Can Regulation Illuminate the Dark Forest?",
booktitle="Advanced Information Systems Engineering Workshops",
year="2023",
publisher="Springer International Publishing",
address="Cham",
pages="186--196",
}

@inproceedings{sokMEV,
author = {Yang, Sen and Zhang, Fan and Huang, Ken and Chen, Xi and Yang, Youwei and Zhu, Feng},
title = {SoK: MEV Countermeasures},
year = {2024},
isbn = {9798400712272},
publisher = {Association for Computing Machinery},
address = {New York, NY, USA},
url = {https://doi.org/10.1145/3689931.3694911},
doi = {10.1145/3689931.3694911},
booktitle = {Proceedings of the Workshop on Decentralized Finance and Security},
pages = {21–30},
numpages = {10},
location = {Salt Lake City, UT, USA},
series = {DeFi '24}
}

@InProceedings{RegMEV,
author="Ji, Yan
and Grimmelmann, James",
editor="Budurushi, Jurlind
and Kulyk, Oksana
and Allen, Sarah
and Diamandis, Theo
and Klages-Mundt, Ariah
and Bracciali, Andrea
and Goodell, Geoffrey
and Matsuo, Shin'ichiro",
title="Regulatory Implications of MEV Mitigations",
booktitle="Financial Cryptography and Data Security. FC 2024 International Workshops",
year="2025",
publisher="Springer Nature Switzerland",
address="Cham",
pages="335--363",
}

@misc{randao,
	author = {Ben Edgington},
	title = {{U}pgrading {E}thereum --- eth2book.info},
	howpublished = {\url{https://eth2book.info/capella/}},
	year = {2024},
	note = {[Accessed 10-07-2025]},
}

@misc{inclusion1,
	author = {Neuder, M. and Buterin, V. and D’Amato, F. and Tsao, T. and Darji, M.},
	title = {{E}{I}{P}-7547: {I}nclusion lists --- eips.ethereum.org},
	howpublished = {\url{https://eips.ethereum.org/EIPS/eip-7547}},
	year = {2024},
	note = {[Accessed 10-07-2025]},
}

@misc{inclusion2,
	author = {Thiery, T. and D'Amato, F. and Ma, J. and Monnot, B. and Tsao, T. and Kaufmann, J. and Song, J.},
	title = {{E}{I}{P}-7805: {F}ork-choice enforced {I}nclusion {L}ists ({F}{O}{C}{I}{L}) --- eips.ethereum.org},
	howpublished = {\url{https://eips.ethereum.org/EIPS/eip-7805}},
	year = {2024},
	note = {[Accessed 10-07-2025]},
}

@misc{garimidi2025transaction,
      title={Transaction Fee Mechanism Design for Leaderless Blockchain Protocols}, 
      author={Pranav Garimidi and Lioba Heimbach and Tim Roughgarden},
      year={2025},
      eprint={2505.17885},
      archivePrefix={arXiv},
      primaryClass={cs.GT},
      url={https://arxiv.org/abs/2505.17885}, 
}

@misc{stouka2025multiple,
      title={Multiple Proposer Transaction Fee Mechanism Design: Robust Incentives Against Censorship and Bribery}, 
      author={Aikaterini-Panagiota Stouka and Julian Ma and Thomas Thiery},
      year={2025},
      eprint={2505.13751},
      archivePrefix={arXiv},
      primaryClass={cs.GT},
      url={https://arxiv.org/abs/2505.13751}, 
}

@INPROCEEDINGS{VRaas2024,
  author={Gorman, Jacob and Hanzlik, Lucjan and Kate, Aniket and Mangipudi, Easwar Vivek and Mukherjee, Pratyay and Sarkar, Pratik and Thyagarajan, Sri AravindaKrishnan},
  booktitle={2025 IEEE 38th Computer Security Foundations Symposium (CSF)}, 
  title={VRaaS: Verifiable Randomness as a Service on Blockchains}, 
  year={2025},
  volume={},
  number={},
  address={Santa Cruz, CA, USA},
  publisher={IEEE},
  pages={331-346},
  doi={10.1109/CSF64896.2025.00034}}

@InProceedings{ourborous,
author="Kiayias, Aggelos
and Russell, Alexander
and David, Bernardo
and Oliynykov, Roman",
editor="Katz, Jonathan
and Shacham, Hovav",
title="Ouroboros: A Provably Secure Proof-of-Stake Blockchain Protocol",
booktitle="Advances in Cryptology -- CRYPTO 2017",
year="2017",
publisher="Springer International Publishing",
address="Cham",
pages="357--388",
isbn="978-3-319-63688-7"
}

@inproceedings{ouroborosGenesis,
author = {Badertscher, Christian and Ga\v{z}i, Peter and Kiayias, Aggelos and Russell, Alexander and Zikas, Vassilis},
title = {Ouroboros Genesis: Composable Proof-of-Stake Blockchains with Dynamic Availability},
year = {2018},
isbn = {9781450356930},
publisher = {Association for Computing Machinery},
address = {New York, NY, USA},
url = {https://doi.org/10.1145/3243734.3243848},
doi = {10.1145/3243734.3243848},
booktitle = {Proceedings of the 2018 ACM SIGSAC Conference on Computer and Communications Security},
pages = {913–930},
numpages = {18},
location = {Toronto, Canada},
series = {CCS '18}
}

@InProceedings{ouroborosPraos,
author="David, Bernardo
and Ga{\v{z}}i, Peter
and Kiayias, Aggelos
and Russell, Alexander",
editor="Nielsen, Jesper Buus
and Rijmen, Vincent",
title="Ouroboros Praos: An Adaptively-Secure, Semi-synchronous Proof-of-Stake Blockchain",
booktitle="Advances in Cryptology -- EUROCRYPT 2018 ",
year="2018",
publisher="Springer International Publishing",
address="Cham",
pages="66--98",
isbn="978-3-319-78375-8"
}

@phdthesis{buchman2016tendermint,
  title={Tendermint: Byzantine fault tolerance in the age of blockchains},
  author={Buchman, Ethan},
  year={2016},
  school={University of Guelph}
}

@misc{rasheed2024evolution,
      title={MEV Ecosystem Evolution From Ethereum 1.0}, 
      author={Rasheed and Yash Chaurasia and Parth Desai and Sujit Gujar},
      year={2024},
      eprint={2406.13585},
      archivePrefix={arXiv},
      primaryClass={cs.CR},
      url={https://arxiv.org/abs/2406.13585}, 
}

@InProceedings{ProofofSpace,
author="Ateniese, Giuseppe
and Bonacina, Ilario
and Faonio, Antonio
and Galesi, Nicola",
editor="Abdalla, Michel
and De Prisco, Roberto",
title="Proofs of Space: When Space Is of the Essence",
booktitle="Security and Cryptography for Networks",
year="2014",
publisher="Springer International Publishing",
address="Cham",
pages="538--557",
isbn="978-3-319-10879-7"
}

@ARTICLE{repuCoin,
  author={Yu, Jiangshan and Kozhaya, David and Decouchant, Jeremie and Esteves-Verissimo, Paulo},
  journal={IEEE Transactions on Computers}, 
  title={RepuCoin: Your Reputation Is Your Power}, 
  year={2019},
  volume={68},
  number={8},
  pages={1225-1237},
  doi={10.1109/TC.2019.2900648}}

@misc{ethereum_validator_phase0,
  title        = {Ethereum Consensus Specs: Phase 0 -- Honest Validator},
  author       = {{Ethereum Foundation}},
  howpublished = {\url{https://github.com/ethereum/consensus-specs/blob/dev/specs/phase0/validator.md}},
  note         = {Accessed: 2025-07-29},
  year         = {2023},
  version      = {dev branch}
}

@misc{blazquez2023relays,
  author       = {Alvaro Blazquez and Sachin Lal},
  title        = {Relays are a Latency Game},
  howpublished = {\url{https://blog.metrika.co/relays-are-a-latency-game-303aadb393ce}},
  note         = {Published by Metrika on Medium. Accessed: 2025-07-29},
  year         = {2023}
}

@inproceedings{micali1999verifiable,
  author    = {Silvio Micali and Michael Rabin and Salil Vadhan},
  title     = {Verifiable Random Functions},
  booktitle = {Proceedings 40th Annual Symposium on Foundations of Computer Science (Cat. No.99CB37039)},
  year      = {1999},
  pages     = {120--130},
  doi       = {10.1109/SFFCS.1999.814584},
  publisher = {IEEE},
  address = {USA}
}

@book{myerson1997,
  author    = {Roger B. Myerson},
  title     = {Game Theory: Analysis of Conflict},
  publisher = {Harvard University Press},
  year      = {1997},
  address   = {Cambridge, MA},
  isbn      = {978-0674341166}
}

@article{nash1950equilibrium,
  title={Equilibrium points in n-person games},
  author={Nash Jr, John F},
  journal={Proceedings of the national academy of sciences},
  volume={36},
  number={1},
  pages={48--49},
  year={1950},
  publisher={national academy of sciences}
}

@inproceedings{karakostas2024blockchain,
author = {Karakostas, Dimitris and Kiayias, Aggelos and Zacharias, Thomas},
title = {Blockchain Bribing Attacks and the Efficacy of Counterincentives},
year = {2024},
isbn = {9798400706363},
publisher = {Association for Computing Machinery},
address = {New York, NY, USA},
url = {https://doi.org/10.1145/3658644.3670330},
doi = {10.1145/3658644.3670330},
booktitle = {Proceedings of the 2024 on ACM SIGSAC Conference on Computer and Communications Security},
pages = {1031–1045},
numpages = {15},
location = {Salt Lake City, UT, USA},
series = {CCS '24}
}

@techreport{buterin2018discouragement,
  author      = {Vitalik Buterin},
  title       = {Discouragement Attacks},
  institution = {Ethereum Foundation},
  year        = {2018},
  note        = {\url{https://eips.ethereum.org/assets/eip-2982/ef-Discouragement-Attacks.pdf}}
}

@article{knight2018nashpy,
  title={Nashpy: A Python library for the computation of Nash equilibria},
  author={Knight, Vincent and Campbell, James},
  journal={Journal of Open Source Software},
  volume={3},
  number={30},
  pages={904},
  year={2018},
  publisher={The Open Journal},
  doi={10.21105/joss.00904},
  url={https://doi.org/10.21105/joss.00904}
}

@book{tadelis2012game,
  title = {Game Theory: An Introduction},
  author = {Tadelis, Steven},
  year = {2012},
  publisher = {Princeton University Press},
  address = {Princeton, NJ},
  isbn = {9780691129082}
}

@misc{sorellaLabs,
	author = {{Sorella Labs}},
	title = {{S}orella {L}abs | {F}air {M}arkets for {A}ll {D}e{F}i {U}sers --- sorellalabs.xyz},
	howpublished = {\url{https://sorellalabs.xyz/}},
	year = {2024},
	note = {[Accessed 01-08-2025]},
}

@InProceedings{gilbert2003security,
author="Gilbert, Henri
and Handschuh, Helena",
editor="Matsui, Mitsuru
and Zuccherato, Robert J.",
title="Security Analysis of SHA-256 and Sisters",
booktitle="Selected Areas in Cryptography",
year="2004",
publisher="Springer Berlin Heidelberg",
address="Berlin, Heidelberg",
pages="175--193",
isbn="978-3-540-24654-1"
}

@ARTICLE{gammap2p2007,
       author = {{Guo}, Tong-qiang and {Weng}, Jian-guang and {Zhuang}, Yue-ting},
        title = "{Content subscribing mechanism in P2P streaming based on gamma distribution prediction}",
      journal = {Journal of Zhejiang University-SCIENCE A},
         year = 2007,
        month = nov,
       volume = {8},
       number = {12},
        pages = {1983-1989},
          doi = {10.1631/jzus.2007.A1983},
       adsurl = {https://ui.adsabs.harvard.edu/abs/2007JZUSA...8.1983G}
}

@techreport{avarikioti2022harmony,
    author = "Avarikioti, Zeta and Karakostas, Dimitris",
    title = "Harmony technical report",
    institution = "Harmony",
    year = {2022}
}

@inproceedings{camenisch2022internet,
  title={Internet computer consensus},
  author={Camenisch, Jan and Drijvers, Manu and Hanke, Timo and Pignolet, Yvonne-Anne and Shoup, Victor and Williams, Dominic},
  booktitle={Proceedings of the 2022 ACM Symposium on Principles of Distributed Computing},
  pages={81--91},
  year={2022}
}

@InProceedings{bentov2016snow,
author="Daian, Phil
and Pass, Rafael
and Shi, Elaine",
editor="Goldberg, Ian
and Moore, Tyler",
title="Snow White: Robustly Reconfigurable Consensus and Applications to Provably Secure Proof of Stake",
booktitle="Financial Cryptography and Data Security",
year="2019",
publisher="Springer International Publishing",
address="Cham",
pages="23--41",
isbn="978-3-030-32101-7"
}

@misc{codexBlockDiffusion,
	author = {Csaba, Kiraly and Leonardo, Bautista-Gomez},
	title = {{B}ig {B}lock {D}iffusion and {O}rganic {B}ig {B}locks on {E}thereum --- blog.codex.storage},
	howpublished = {\url{https://blog.codex.storage/big-blocks-on-mainnet/}},
	year = {2023},
	note = {[Accessed 25-09-2025]},
}

@InProceedings{ethGossip2021,
author="Kiffer, Lucianna
and Salman, Asad
and Levin, Dave
and Mislove, Alan
and Nita-Rotaru, Cristina",
editor="Borisov, Nikita
and Diaz, Claudia",
title="Under the Hood of the Ethereum Gossip Protocol",
booktitle="Financial Cryptography and Data Security",
year="2021",
publisher="Springer Berlin Heidelberg",
address="Berlin, Heidelberg",
pages="437--456",
isbn="978-3-662-64331-0"
}

@inproceedings{ijcai2024,
  title     = {Game Transformations That Preserve Nash Equilibria or Best-Response Sets},
  author    = {Tewolde, Emanuel and Conitzer, Vincent},
  booktitle = {Proceedings of the Thirty-Third International Joint Conference on
               Artificial Intelligence, {IJCAI-24}},
  publisher = {International Joint Conferences on Artificial Intelligence Organization},
  editor    = {Kate Larson},
  pages     = {2984--2993},
  year      = {2024},
  month     = {8},
  note      = {Main Track},
  doi       = {10.24963/ijcai.2024/331},
  url       = {https://doi.org/10.24963/ijcai.2024/331},
  address = {Jeju, Korea}
}

@inproceedings{BitcoinF2020,
author = {Siddiqui, Shoeb and Vanahalli, Ganesh and Gujar, Sujit},
title = {BitcoinF: Achieving Fairness For Bitcoin In Transaction Fee Only Model},
year = {2020},
isbn = {9781450375184},
publisher = {International Foundation for Autonomous Agents and Multiagent Systems},
address = {Richland, SC},
pages = {2008–2010},
numpages = {3},
location = {Auckland, New Zealand},
series = {AAMAS '20}
}

@inproceedings{NFairness2021,
  author       = {Anurag Jain and
                  Shoeb Siddiqui and
                  Sujit Gujar},
  editor       = {Frank Dignum and
                  Alessio Lomuscio and
                  Ulle Endriss and
                  Ann Now{\'{e}}},
  title        = {We might walk together, but {I} run faster: Network Fairness and Scalability
                  in Blockchains},
  booktitle    = {{AAMAS} '21: 20th International Conference on Autonomous Agents and
                  Multiagent Systems, Virtual Event, United Kingdom, May 3-7, 2021},
  pages        = {1539--1541},
  publisher    = {{ACM}},
  year         = {2021},
  url          = {https://www.ifaamas.org/Proceedings/aamas2021/pdfs/p1539.pdf},
  doi          = {10.5555/3463952.3464152},
  bibsource    = {dblp computer science bibliography, https://dblp.org}
}

@inproceedings{heimbach2022,
author = {Heimbach, Lioba and Wattenhofer, Roger},
title = {Eliminating Sandwich Attacks with the Help of Game Theory},
year = {2022},
isbn = {9781450391405},
publisher = {Association for Computing Machinery},
address = {New York, NY, USA},
url = {https://doi.org/10.1145/3488932.3517390},
doi = {10.1145/3488932.3517390},
booktitle = {Proceedings of the 2022 ACM on Asia Conference on Computer and Communications Security},
pages = {153–167},
numpages = {15},
location = {Nagasaki, Japan},
series = {ASIA CCS '22}
}

@inproceedings{braga2024,
author = {Braga, Pedro and Chionas, Georgios and Krysta, Piotr and Leonardos, Stefanos and Piliouras, Georgios and Ventre, Carmine},
title = {Who gets the Maximal Extractable Value? A Dynamic Sharing Blockchain Mechanism},
year = {2024},
isbn = {9798400704864},
publisher = {International Foundation for Autonomous Agents and Multiagent Systems},
address = {Richland, SC},
booktitle = {Proceedings of the 23rd International Conference on Autonomous Agents and Multiagent Systems},
pages = {2171–2173},
numpages = {3},
location = {Auckland, New Zealand},
series = {AAMAS '24}
}

@inproceedings{Rasheed2025,
author = {Rasheed and Desai, Parth Nimish and Chaurasia, Yash and Gujar, Sujit},
title = {Shapley Value-based Approach for Distributing Revenue of Matchmaking of Private Transactions in Blockchains},
year = {2025},
isbn = {9798400714269},
publisher = {International Foundation for Autonomous Agents and Multiagent Systems},
address = {Richland, SC},
booktitle = {Proceedings of the 24th International Conference on Autonomous Agents and Multiagent Systems},
pages = {2723–2725},
numpages = {3},
location = {Detroit, MI, USA},
series = {AAMAS '25}
}

@inproceedings{chitra2022,
author = {Chitra, Tarun and Kulkarni, Kshitij},
title = {Improving Proof of Stake Economic Security via MEV Redistribution},
year = {2022},
isbn = {9781450398824},
publisher = {Association for Computing Machinery},
address = {New York, NY, USA},
url = {https://doi.org/10.1145/3560832.3564259},
doi = {10.1145/3560832.3564259},
booktitle = {Proceedings of the 2022 ACM CCS Workshop on Decentralized Finance and Security},
pages = {1–7},
numpages = {7},
location = {Los Angeles, CA, USA},
series = {DeFi'22}
}
